\documentclass[12pt]{article}
\pdfoutput=1
\usepackage{jheppub}
\usepackage[T1]{fontenc}
\usepackage{hyperref}
\usepackage{comment}
\usepackage{outlines}
\usepackage{setspace}
\usepackage{cancel}
\usepackage{stmaryrd}
\usepackage{amsthm}
\usepackage{enumerate}
\usepackage[normalem]{ulem}
\title{Dressed Subsystems in Classical Gravity}
\author{Pranav Pulakkat}
\affiliation{Maryland Center for Fundamental Physics,\\ University of Maryland, College Park, MD 20740, USA}
\emailAdd{pranavp@umd.edu}
\date{August 2024}
\newtheoremstyle{breakthm}
{}{}
{\itshape}{}
{\bfseries}{}
{\newline}{}

\theoremstyle{breakthm}
\newtheorem{prop}{Proposition}[section]
\newtheorem{cor}{Corollary}

\newtheoremstyle{break}
{\topsep}{\topsep}
{}{}
{\bfseries}{}
{\newline}{}

\theoremstyle{break}
\newtheorem*{defn}{Definition}
\newtheorem*{rem}{Remark}
\newtheorem*{claim}{Claim}
\newtheorem{assumption}{Assumption}

\newcommand{\Ss}{\mathcal{S}(\Sigma)}
\newcommand{\Ssb}{\mathcal{S}(\overline{\Sigma})}
\newcommand{\Ssp}{\mathcal{S}^{\prime}(\Sigma)}
\newcommand{\Ssbp}{\mathcal{S}^{\prime}(\overline{\Sigma})}
\newcommand{\Sspg}{\mathcal{S}_{g}^{\prime}(\Sigma)}
\newcommand{\Ssbpg}{\mathcal{S}_{g}^{\prime}(\overline{\Sigma})}
\abstract{This paper considers the problem of consistently defining subsystems in gravitational theories. It is argued that a subsystem is a spacetime subregion in which the observables form a closed Poisson algebra. In a generally covariant theory, the location of the subregion must be determined in relation to other degrees of freedom. It is proposed that these degrees of freedom should live within the region, so that an observer can determine its edge by only measuring fields inside of it. This turns out to be equivalent to the property that observables in the subregion generate field-dependent gauge transformations on the causal complement. Furthermore, it is demonstrated that this is exactly what is needed for the observables to form a Poisson algebra and thus to constitute a consistent subsystem. Observables in spacelike separated "dressed subsystems" are shown to commute. Several examples are given in the context of General Relativity. Along the way, new perspectives on the covariant phase space formalism are introduced that clarify well-known issues, such as the factorization of subregions in gauge theories and the unambiguous definition of Noether charges associated with one-sided boosts. Finally, prospects for extending these results to a perturbative quantum setting are discussed.}

\begin{document}
\flushbottom
\maketitle

\section{Introduction}
The concept of a subsystem is fundamental to theoretical physics. In order for a theory to have practical predictive power, it must be possible to identify localized sets of degrees of freedom that evolve in a self-contained way. Otherwise an observer would need to collect an enormous amount of delocalized data in order to make any predictions about the future. Our ability to experimentally test theories therefore depends upon on our ability to define local subsystems and extract descriptions of their dynamics. 
\par
In quantum field theory, subsystems are defined as causally complete regions of spacetime. The degrees of freedom belonging to a region consist of the operators that can be built out of local field operators in that region. A subsystem should have most of the properties of a complete physical system; in particular, the operators belonging to a subsystem should form a closed algebra. This is the essence of the algebraic approach to quantum field theory \cite{Haag}, which formalizes these algebras as $C^*$ or von Neumann algebras. In fact, the way that subsystem algebras are nested in one another may encode all physical content of the theory.
\par
A quantum description of gravity must also admit a notion of subsystem if it is to be useful, at least within the regime of effective field theory. There are two obstacles to formulating a similar, algebraic characterization of subsystems when gravity is dynamical. First, the metric fluctuates, so there is no fixed notion of a causally complete region. Second, gravitational theories are generally covariant, so observables must be invariant under diffeomorphisms that are supported in the bulk. This means that there are \textit{no} observables supported in a compact region. Together these issues obviate the possibility of defining a subsystem in terms of fixed regions on a background manifold.
\par
There is, however, a well-known solution to the second problem. Observables may appear to be completely delocalized with respect to the spacetime manifold, but can localize \textit{in relation} to other dynamical degrees of freedom. This concept of a relationally local observable dates back to Bergmann and Komar \cite{Komar,Bergmann_Komar,Bergmann_1, Bergmann_2}, and since then has been the subject of an enormous literature. For a small and unsystematic selection, see \cite{Bianca,Chataignier,khavkine_observable,frames,reference,reference_frames,Marolf_causality,Z-model,Kaplan_dS}. This means that in order for a subsystem to contain proper observables, its location must also be defined relationally. Outside of special cases \cite{reference_frames, Marolf_causality, Kaplan_dS, locally_covariant_gravity}, this is not well studied. In particular, the conditions under which a relational specification of a subregion defines a proper subsystem, described by a closed algebra of observables, are not understood. 
\par
In this paper, we study this problem within the context of classical field theory using the covariant phase space formalism. The corresponding notion of algebraic closure is the closure of the Poisson brackets among subregion observables into an algebra. For a general relational specification of a subregion, this property will not hold. However, we will find there is an extremely natural criterion that is exactly sufficient to guarantee that it does; namely, the requirement that the degrees of freedom locating the region also live within the region. In other words, an observer would be able to determine the edge of the region by only taking measurements inside of it. We refer to such regions as \textit{internally dressed subsystems}. Donnelly and Giddings have established within perturbation theory that in general, relational observables will not commute at spacelike separation \cite{dress,obstruction_subsystem}. They interpreted this as an obstruction to a proper definition of local subsystems in gravitational theories. I argue that this is a premature conclusion; as long as the internal dressing criterion is obeyed, there is no issue with interpreting a relationally local subregion as a subsystem. While observables within that region will generate a nontrivial action on the causal complement, this action is highly constrained to consist of a \textit{field-dependent gauge transformation}. In fact, this property is crucial for proving the closure of the algebra in the first place.
\\
\indent Our arguments will rely on an extension of the covariant phase space formalism, explained in section \ref{subsec:prelim}. In particular, we will need to postulate two causality properties of the linearized equations of motion. These are automatically satisfied when the linearized equations can be expressed in a hyperbolic gauge. In fact these properties are precisely those that recast the covariant Poisson bracket as a  Peierls bracket \cite{peierls} well-defined. There is a deep relation between subsystem algebras and the Peierls bracket, which will be intertwined throughout our discussion. The bracket also motivates the introduction of distributional flows on the phase space, associated with singular observables.
\\
\indent In section \ref{subsec:subsystem}, in order to lay out a blueprint for the study of dressed subsystems, we will prove various properties for subsystems in theories without dynamical gravity or gauge invariance, including the closure of the Poisson algebra. After taking a detour to analyze the structure of gauge transformations in a theory that satisfies the causality assumptions (\ref{subsec:gauge_character}, Appendix \ref{app:gauge}), we will extend this analysis to nongravitational gauge theories (\ref{subsec:gauge_subsystem}). In section \ref{subsec:dress} we will introduce the concept of an internally dressed subsystem and reformulate it into a more convenient language via gauge fixing. This prepares us for the analysis of \ref{subsec:gauge_fix}, which establishes the claimed results about the action of subsystem observables on the causal complement and their closure into an algebra. In \ref{sec:examples}, we briefly look at some examples. In the discussion, we will see some other insights that come out of our methods, regarding the independence of subregions and the factorization of gauge theory phase spaces (\ref{subsec:factor}), and the implementation of surface charges for subregions (\ref{subsec:surface}). We conclude by returning to the question of how to define subsystems in the quantum theory, and sketch how the present results might be reformulated quantum mechanically. In order to make clear the logical structure of this framework, the paper is structured in proposition-proof format. This should not be misinterpreted as implying complete mathematical rigour. Although we will pay attention to certain subtleties that are usually ignored, the attitude will still be pragmatic and based on physical interpretation.
\subsection{Conventions}\label{subsec:conven}
 Our conventions mostly follow those of \cite{harlow}. The mostly plus metric signature is employed, with total spacetime dimension $D$. The covariant phase space formalism makes use of differential forms that carry indices on both spacetime and the space of field histories. These both obey the same conventions for the exterior derivative and wedge product:
\begin{equation}\label{forms}
    \begin{aligned}
        (\alpha\wedge\beta)_{a_1...a_pb_1...b_q} &= \frac{(p+q)!}{p!q!}\alpha_{[a_1...a_p}\beta_{b_1...b_q]}\\
        d\alpha_{ca_1...a_p} &=(p+1)\partial_{[c}a_{a_1...a_p]};
    \end{aligned}
\end{equation}
however the configuration space wedge product is implicit whenever two forms are multiplied by each other. The spacetime exterior derivative is represented by $d$ and the configuration space exterior derivative, also called the \textit{variation}, is denoted $\delta$. Contraction of a vector $X$ with the first spacetime index is denoted with a lowercase $i_X$, and contraction of a configuration space vector (also called a \textit{flow}) $\chi$ with the first configuration space index is denoted with an uppercase $I_{\chi}$. We will not be careful to distinguish presymplectic and properly symplectic quantities in terminology. 
\section{Revisiting Covariant Phase Space}
\subsection{Preliminaries}\label{subsec:prelim}
Our fist order of business is to clarify some basic issues that arise when applying the covariant phase space method to the analysis of subsystems. We will start by reviewing the formalism along the lines presented in \cite{harlow}. Consider a $D$-dimensional oriented spacetime manifold $\mathcal{M}$, on which there is defined a number of physical fields. These can be taken locally to be smooth sections of some vector bundles over $\mathcal{M}$. Some of these may be background fields, which are held constant; others will be dynamical, and will be allowed to vary. The \textit{configuration space} $\mathcal{C}$ is the space of all possible histories of the dynamical fields, i.e. the space of sections of the relevant bundles. 
\par
The covariant phase space method assumes that this is an infinite-dimensional manifold, to which standard techniques of differential topology and calculus can be applied. Properly speaking this requires $\mathcal{C}$ to have a sufficiently well-behaved structure; naturally it belongs to the category of Frechet manifolds \cite{functional_spaces, khavkine_observable}. This permits the construction of tangent and cotangent bundles, finite tensor powers thereof, and the usual operations of the Lie and exterior derivative, which are related by Cartan's magic formula (consult \ref{subsec:conven} for notation). However, \textit{not} all results carry over from the finite dimensional setting. For example, the Frobenius theorem is not valid for arbitrary distributions on Frechet manifolds \cite{convenient}. In certain cases these subtleties will turn out to have an important physical interpretation. As the need arises, we will evaluate their relevance with respect to field theories of interest.
\begin{rem}
    The requirement of working on a Frechet space comes from the assumption that field configurations are smooth. If we were to instead consider configurations differentiable only up to some finite order, $\mathcal{C}$ would be a Banach manifold. These obey a more familiar theory; in particular, the Frobenius theorem \textit{does} hold \cite{lang}. From a physical point of view, as long as the derivative order is sufficiently high, this should be just as good of a setting for covariant phase space analysis. However, there is a tradeoff that makes it much less convenient when discussing gravitational theories. When only a finite degree of differentiability is required, the infinite dimensional symmetry groups that act on the configuration space (such as the diffeomorphism group) are typically not Lie groups \cite{convenient}. This means that transformations cannot always be connected by differentiable paths. Since our methods will involve infinitesimal gauge transformations, it is much more tractable to work in the Frechet context and do without the general validity of the Frobenius theorem.
\end{rem}

\par
The Lagrangian $\mathcal{L}$ is a spacetime top form, configuration space scalar locally constructed out of the dynamical and background fields, as well as their derivatives. Note that any dependence on the spacetime point must come through the fields, and that there is no background structure except for the specified fields. This means, for example, that any derivative operators appearing in the Lagrangian must be determined in terms of the fields. Taking the configuration space exterior derivative of $\mathcal{L}$ and integrating by parts, we obtain
\begin{equation}\label{covariant}
    \delta \mathcal{L}=E_a\delta\phi^a+d\theta
\end{equation}
where the contracted index sums over the different physical fields. The equations of motion are the spacetime top forms $E_a$, and the prephase space $\mathcal{P}$ is defined to be the submanifold of configuration space where these all vanish, which is also assumed to be Frechet and path-connected. The symplectic potential $\theta$ is a degree $D-1$ form locally constructed out of the fields, the variations $\delta\phi^a$, and derivatives thereof, and is linear in the variations. It is defined up to the addition of a likewise closed, locally constructed form. We have the following theorem proven by Wald in \cite{Wald_locally_constructed}, and also known in the BRST literature and separately proven \cite{local_BRST_cohomology,completeness_BRST}:

\begin{prop}[Algebraic Poincare Lemma]\label{prop:apl}
Any locally constructed $p$ form, where $p<D$ that is identically closed --- closed for arbitrary sections of the vector bundles corresponding to the dynamical fields --- is the exterior derivative of some locally constructed differential $p-1$ form.
\end{prop}
Treating $\delta\phi^a$ as dynamical fields, it follows that for any allowed $\theta$, $\theta+d\beta$ is also allowed, where $C$ is a $D-2$ form locally constructed out of the fields and $\delta\phi^a$, and linear in the latter \cite{JKM}. \footnote{There is an alternative viewpoint based on the "Anderson homotopy operator", which amounts to a specific systematic prescription for defining $\theta$ from the Lagrangian. The result of this recipe is changed if one adds a spacetime exact term to the Lagrangian, exactly by a locally exact term as above. Regardless of whether one takes the view that the choice of boundary term in the Lagrangian is irrelevant and there is a freedom in choosing the symplectic potential, or that the choice of boundary term is physically important and the symplectic potential is then uniquely determined, there is no material difference; there is a choice involved which preserves $E_a$ but may alter $\theta$ by an exact term.}
\par
\begin{rem}
    Note that a local field redefinition that mixes with derivatives, i.e. $\phi'=\phi+u^a\partial_a\phi+...$ will change the definitions of both $E_a$ and $d\theta$. However by integrating by parts it is easy to see that these corrections vanish on-shell, so that the solution space and the pullback to it of $d\theta$ is unchanged. 
\end{rem}

The symplectic current is defined as 
\begin{equation}
    \omega=\delta\theta,
\end{equation}
and integrating it over some codimension-1 hypersurface defines a symplectic form $\Omega$ on prephase space (properly speaking, it is only presymplectic due to the possible presence of degenerate directions --- however, we will not use this terminology). Note that, just like $\theta$, it is defined only up to the addition of an exact term which will affect the total symplectic form by a term localized to the codimension 2 corner that bounds the surface; however, it is completely unaffected by the addition of an exact form to the Lagrangian.  It is well known that the pullback of $\omega$ to the prephase space is conserved, i.e.
\begin{equation}\label{conserve}
    d\omega(X,Y)=0
\end{equation}
when $X$ and $Y$ are tangent to prephase space. This means that integrating $\omega$ over homologous surfaces will yield the same result for $\Omega$. The question of what class of surfaces to integrate over involves issues of causality that are usually glossed over in the literature. However, since we will need to be careful about the causal properties of our theory at several points, let us enumerate these now.
\par
We suppose that one of the fields, either background or dynamical, is a Lorentzian metric $g$ with respect to which the spacetime manifold is globally hyperbolic for \textit{all} configurations. Furthermore, $g$ is taken to determine the causal structure. However, we do not simply say that "the equations of motion are hyperbolic with respect to the metric" for two reasons. First, gauge redundancy may cause the equations of motions to fail to provide a unique solution for an arbitrary choice of data prescribed on an initial surface. Second, if the metric is dynamical, the causal structure will depend on the state and cannot determine fixed properties of the equations of motion $E_a$. Instead, we will only constrain the linearized equations of motion about any particular solution, which are encoded by the forms $\delta E_a$. These determine the propagation of linearized perturbations to the dynamical fields on a fixed background. They may still not be well-posed, but the nonuniqueness issue will be handled by analyzing gauge redundancies. These are related to Noether identities: differential equations that are identically obeyed by the linearized equations, which can be denoted schematically as $\Delta^j(\delta E_a)=0$. We will only properties that guarantee the local existence of solutions given appropriate initial data and sources. To wit:

\begin{assumption}\label{causal_1}
For a given solution, consider an arbitrary partial surface $\Sigma$ and its domain of dependence $D(\Sigma)$, defined in the standard way with respect to $g$. Data can be specified consisting of the field values $\delta\phi^a$ and some number of transverse derivatives thereof on $\Sigma$, as well as smooth sources $J^a$  associated to each dynamical field $\phi^a$, defined throughout $D(\Sigma)$. These are arbitrary  up to the Noether identities $\Delta^j(J_a)=0$ and local differential constraints $C^k(\delta\phi^a,J_a)$ on the data on $\Sigma$. Any such data set can be extended to a smooth solution $\delta\phi_a$ of the inhomogeneous linearized system $\delta E_a=J_a$, throughout all of $D(\Sigma)$. Conversely any inhomogeneous solution $(\delta\phi^a,J_a)$ of this system induces such a data set.
\end{assumption}

\begin{assumption}\label{causal_2}
  Let $\Sigma$ be partitioned into two partial Cauchy surfaces. Consider \textit{homogeneous} linearized solutions $\delta phi^a$ on the domains of dependence of each partial surface that agree to all derivatives at the codimension-2 surface where the regions meet. Then there exists a solution on the entirety of $D(\Sigma)$ that restricts to the respective solution on each domain.
\end{assumption}
\begin{figure}
    \centering
    \includegraphics[scale=.2]{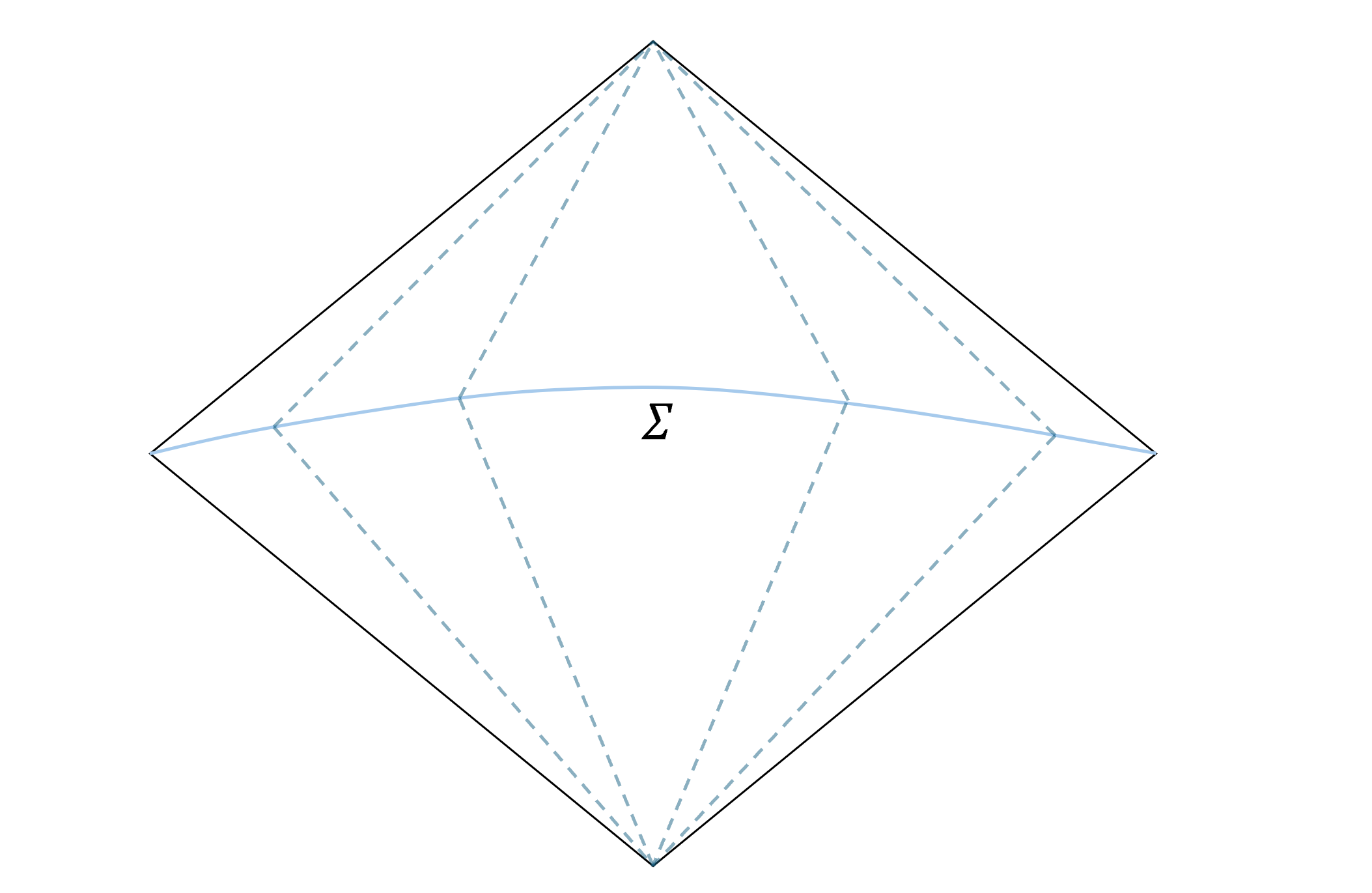}
    \caption{Data for the linearized equations on $\Sigma$ can be propagated throughout $D(\Sigma)$, as shown by the dotted lines.}
    \label{fig:existence}
\end{figure}

\begin{figure}
    \centering
    \includegraphics[scale=.2]{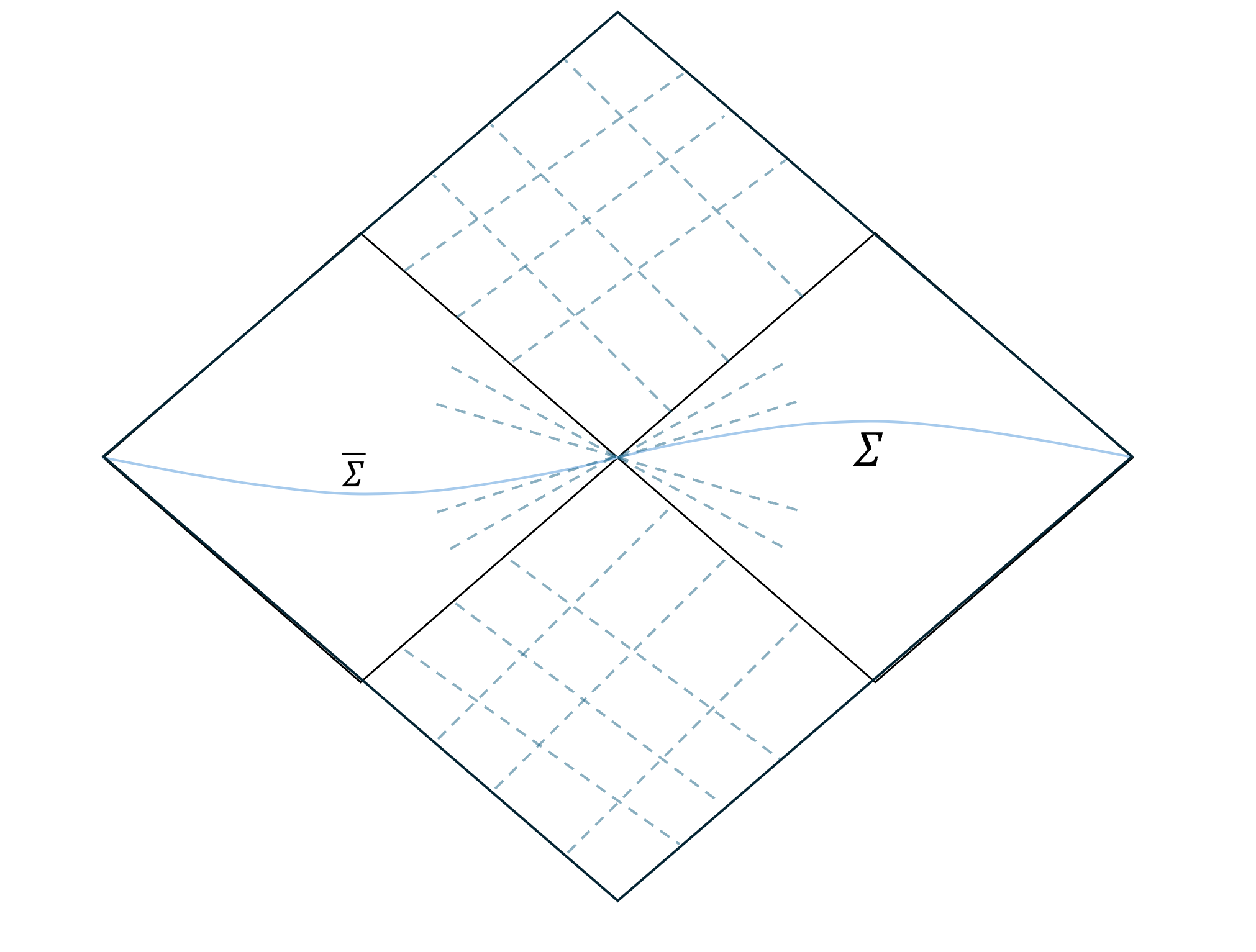}
    \caption{Flows in $D(\overline{\Sigma}$) and $D(\Sigma$) that agree smoothly near $\partial\Sigma$ can be extended to a flow throughout the common domain of dependence.}
    \label{fig:gluing}
\end{figure}

\par
It may seem odd that the gluing property of Assumption \ref{causal_2} refers only to the homogeneous equations, but Assumption \ref{causal_2} needs to discuss the inhomogeneous problem\footnote{I thank Aldo Riello and Rodrigo Andrade e Silva for discussion on the necessity of  inhomogeneous existence. A previous version of this article on arXiV only emphasized the assumption of homogeneous existence, but implicitly assumed this form of inhomogeneous existence in Appendix \ref{app:gauge}. A somewhat similar concept appears under the name of local completeness of the Noether identities in \cite{antifield_lectures}.}. This turns out to be necessary to have a good Peierls bracket formulation, as we will discuss shortly. It also plays a central role in the analysis of gauge symmetries carried out in Appendix \ref{app:gauge} and used in section \ref{subsec:gauge_character}. For now, note that both axioms hold for linearized equations of motion that can, with some choice of a local gauge fixing, be written as an inhomogeneous PDE system that is hyperbolic with respect to $g$ and well-posed for sources satisfying the Noether identities. Here, "gauge fixing" is meant in the sense of PDE theory, and is so far unrelated to any symplectic structure. To be precise, we mean that there is a class of smooth invertible transformations that can be applied to a linearized configuration on $D(\Sigma)$ in order to impose an additional set of local differential equations, which are the gauge fixing conditions.
\par
Of course, this applies to pure electromagnetism and Einstein gravity; the linearization of the former can be put into Lorenz gauge, and the latter into a de Donder gauge. These equations admit retarded Green's functions. For a source obeying the Noether identities, integrating the source against the Green's function and adding a suitable homogeneous solution can produce a solution with the prescribed initial data and source, verifying Assumption \ref{causal_1}. In the case of electromagnetism the Noether identities require the source to be conserved. This is consistent with the Lorenz gauge condition for conserved sources. These properties can be straightforwardly extended to linearized Yang-Mills and Einstein gravity.
\par
Turning to Assumption \ref{causal_2}, given two homogeneous solutions in the complementary domains of dependence that agree to all orders at the intersecting surface, each one can be transformed to the hyperbolic gauge separately. However, since the gauge fixing conditions are local, the transformed solutions will also agree to arbitrary order near the intersecting surface. The data sets for these solutions glue into a single data set on $\Sigma$ that can be evolved to a complete solution on $D(\Sigma)$ in the hyperbolic gauge. The well-posedness of the gauge-fixed Cauchy problem implies that this solution must agree with the transformed solutions in either of the smaller domains of dependence. Now the entire solution can be transformed back to a solution that agrees with the original solutions on each of the complementary domains, as desired. Although the hyperbolicity condition can be made manifest in practical examples, for the purposes of theoretical reasoning it is much more elegant (as well as in principle more general) to directly employ the stated assumptions, as these do not require any choice of gauge fixing conditions and apply to the original linearized equations. 
\par
We return to the global question of which surfaces determine the symplectic form. The preceding assumptions make it natural to consider Cauchy surfaces, since any initial data set satisfying the constraints induces a linearized solution on the entirety of $\mathcal{M}$ and vice versa. Of course this is not a unique correspondence, but since the symplectic current is locally constructed, only derivatives of the linearized fields at the Cauchy surface contribute to the symplectic form. For the purpose of computing $\Omega$ between tangent vectors at a given point of $\mathcal{P}$, it therefore suffices to parametrize the linearized solutions by the fields up to some derivative order on the surface. This defines a vector space which not subject to any constraints except for the aforementioned local ones.
\par
However, not all Cauchy surfaces are homologous. If the spacetime has a nontrivial timelike (possibly asymptotic) boundary for example, surfaces that intersect this boundary at different places will in general yield different symplectic forms. A resolution to this problem was proposed in \cite{Compere_Marolf_corner} and worked out systematically in \cite{harlow}. The idea is to impose boundary conditions and choose the corner term in the symplectic potential such that the flux of symplectic current through the boundary vanishes. This guarantees that the total symplectic form is conserved, i.e. independent of the Cauchy slice used to compute it. For this paper, it is not important exactly what choices are made to do this; we simply suppose that some prescription has been applied to guarantee conservation of $\Omega$. 
\par
Now that the symplectic form has been identified, its null directions on $\mathcal{P}$ are defined to be the infinitesimal gauge transformations. It is easy to see that these form an involutive distribution; the Lie bracket of any two null vector fields is another null vector field. However, as noted previously the Frobenius theorem is \textit{not} valid on Frechet manifolds, so it does not follow that $\mathcal{P}$ is  foliated by gauge orbits. Nevertheless, in all field theories that are of physical interest, it appears that the gauge distribution \textit{is} integrable and foliates the prephase space. Moreover, the reduced phase space $\tilde{\mathcal{P}}$ is obtained by taking the quotient by the null orbits, and $\Omega$ descends to a nondegenerate symplectic form upon it. Observables are differentiable functions of $\mathcal{P}$ invariant on these orbits, that thus descend to $\tilde{\mathcal{P}}$. We would like to define a Poisson bracket in the usual way; an observable $\mathcal{O}$ is said to generate a flow $\chi$, which is a vector field tangent to prephase space,
\begin{equation}\label{flow}
    \delta \mathcal{O} = \Omega(\cdot,\chi), 
\end{equation}
and then the bracket with another observable $R$ is
\begin{equation}\label{poisson_bracket}
\{\mathcal{R},\mathcal{O}\}=-\delta R(\chi).
\end{equation}
Of course, $\chi$ can only be defined by equation \eqref{flow} up to the addition of a gauge transformation, but this does not affect the Poisson bracket in equation \eqref{poisson_bracket}.

\subsection{Peierls Brackets and Singular Observables}\label{subsec:singular}
This Poisson bracket can be reformulated in terms of the inhomogeneous problem for the linearized equations $\delta E_a$ referring to the symplectic structure. Following the simple argument given in \cite{harlow}, observe that 

\begin{equation} \label{peierls}
    d\omega = d\delta\theta=\delta d\theta=\delta(\delta \mathcal{L}-E_a\delta\phi^a)=\delta\phi^a\delta E_a
\end{equation}

where we used that $d$ and $\delta$ commute and that $\delta^2=0$ (recall that the configuration space exterior product is implicit in the last expression). Now consider a Cauchy slice $\Sigma_1$ and another to its future $\Sigma_2$, and an observable $\mathcal{O}$ supported in the timeslab between them. Suppose that $\frac{\delta O}{\delta\phi^a} =J_a$ for some source $J_a$, and suppose that $J_a$ satisfies the Noether identities. For the moment this is just assumed; we will see in section \ref{subsec:gauge_character} that it is actually equivalent to symplectic gauge invariance when Assumptions \ref{causal_1} and \ref{causal_2} hold. Consider the linearized solution $X_R$ with initial data equal to zero on $\Sigma_1$ and source $J_a$. This is essentially the "retarded effect" of driving the equations of motion with source $J_a$, or equivalently, of subtracting $\mathcal{O}$ from the action. Inserting $X_R$ into \eqref{peierls}, we find

\begin{equation}
    d\omega(\cdot,X_R)=  \delta \phi^a J_a-X_R^a\delta E_a.
\end{equation}
Integrating this in the timeslab between $\Sigma_1$ and $\Sigma_2$ and inserting any \textit{on-shell} linearized flow $Y$ in the other slot gives

\begin{equation}
    \int_{\Sigma_2} \omega (Y,X_R)=\int_{\Sigma_1}^{\Sigma_2}d\omega(Y,X_R)=\int_{\Sigma_1}^{\Sigma_2} [\delta \phi^a(Y) J_a-X_R^a\delta E_a(Y)]=\int_{\Sigma_1}^{\Sigma_2}  Y^a J_a=\delta \mathcal{O} (Y).
\end{equation}
Therefore $X_R$ restricted to the future of $\mathcal{O}$ is equal to the on-shell flow $X$ generated by $\mathcal{O}$ (of course, up to gauge). Likewise we could consider $X_A$ which is a solution with the same source that vanishes on $\Sigma_2$; it will equal $-X$ to the past of $\Sigma_1$. The difference $X_R-X_A$ is a linearized solution that is therefore gauge equivalent to $X$. The Peierls bracket simply defines 
\begin{equation}
    \{\mathcal{R},\mathcal{O}\}=-\delta \mathcal{R}(X_R-X_A)
\end{equation}
which does not explicitly refer to the symplectic structure. When the linearized equations of motion can be put into convenient hyperbolic gauges $X_R$ and $X_A$ can often be constructed in terms of retarded and advanced Green's functions respectively. We did not need to assume such an explicit description. Note that since we have constructed a flow for $\mathcal{O}$, it must be a gauge-invariant observable. Hence the Noether identities for $J_a$ do imply gauge invariance of the linearized observable. The converse will also turn out to be true in section \ref{subsec:gauge_character}. We have not claimed that any observable that generates a symplectic flow can be given a well-defined construction of $X_R$ and $X_A$ that would reproduce its flow. The source defined by its functional derivative could be a singular distribution to which Assumption \ref{causal_1} does not apply as written. An example of this is a spatial smearing of a Klein-Gordon field. This issue may be alleviated by extensions of Assumption \ref{causal_1} to more general distributional sources and solutions.
\par
In fact in field theory, for any given observable there may not even exist $\chi$ satisfying \eqref{flow}. This issue is absent when the reduced phase space is finite dimensional, since then the nondegenerate symplectic form  can be inverted to construct $\chi$ explicitly as a vector tangent to $\tilde{\mathcal{P}}$, which can then be lifted to $\mathcal{P}$. In the infinite dimensional case, the symplectic form may not be invertible even though it is nondegenerate, in which case $\tilde{\mathcal{P}}$ is \textit{weakly symplectic} \cite{convenient}. This is generically the case in field theories. For example, consider Klein-Gordon theory on an arbitrary background. While smearing the field with a smooth spacetime function --- or even a function supported only on a spatial slice --- yields an observable that generates a smooth flow, the field \textit{at a single point} does not. 
\par
The Peierls bracket construction, however, suggests that the flow can still be defined in some sense. For the field at a point $X_R$ and $X_A$ can still be defined as the retarded and advanced Green's functions for a delta function source at that point. These are now \textit{distributions}, not smooth functions. Their difference is another distribution, the Pauli-Jordan commutator kernel \cite{Pauli-Jordan}, which may be "distributional flow" associated to the field. 
\par
There is a natural way to extend the symplectic interpretation of \eqref{flow} to a wider class of observables to accomodate such flows (although we do not claim that it is possible for \textit{all} observables).  This is to \textit{define} $\chi$ to act by dualization on the smooth flow inserted into the other slot of the symplectic form, to yield the variation under that flow of $\mathcal{O}$. In order to do this, one must make sense of the expression $-I_\chi \omega$ as an object that can be integrated against smooth sections of the field bundles on a codimension-1 surface $\Sigma$. Fix some volume form $\mathcal{V}$ and allow $\chi$ to be the adjoint with respect to $\mathcal{V}$ of a distribution dual to the space of smooth sections of the dynamical field vector bundle (not to be confused with the previous use of the word distribution to mean a subbundle of a tangent bundle). In other words, $\chi$ is the formal symbol representing the distribution as a "generalized function". For instance, a pointlike delta distribution with respect to the measure  defined by $\mathcal{V}$ has as its symbol $\delta_\mathcal{V}(x)$.  Observe that in the expression $\omega(\delta_1 \phi, \delta_2 \phi)$, we can formally replace $\delta_2\phi$ with the distributional symbol $\chi$, following the rule that all derivatives are taken to act on variations and not vice versa. This yields a degree $(D-1)$ differential form with coefficients formed out of distributions as well as the flow $\delta_1 \phi$, which is still assumed to be smooth. Differential forms with distributional coefficients were introduced long ago in mathematics by de Rham \cite{deRham,deRham_book}, who conveniently named them \textit{currents}.  There is a well-developed theory of currents, that defines their integration on most surfaces of appropriate dimension and generalizes the standard theorems of exterior calculus, including Stokes theorem. A convenient mathematical reference is \cite{Simon_currents}. The important facts are as follows:
\par
A distribution is a $D-$current; for instance, $\delta_\mathcal{V}(x)\mathcal{V}$ is meant to be integrated against a smooth function over regions of $\mathcal{M}$. Despite this, we will often refer to the formal symbols of the previous paragraph as distributions for convenience. To integrate a current $\psi$ of degree $p$ over a submanifold $H$ of dimension $p$, form $\psi\wedge(\bigwedge_i\delta(f_i)n_i)$, where $f_i$ are $D-p$ independent functions that specify $H$ as their mutual zero set and $n_i$ are their gradients. This object is a $D-$current and therefore a distribution; its integral over $M$ is defined to be the integral of $\psi$ over $H$. This definition makes sense as long as $\psi$ is not singular on $H$, so that the above wedge product is well-defined. For example, if we consider $\psi = \delta_\mathcal{V}(f) U$, where $f$ is a function with a codimension-1 zero set and $U$ is some smooth $D-1$ form, this can be integrated over any submanifold disjoint from or transversely intersecting $f^{-1}(0)$. However it has ill-defined pullback to $f^{-1}(0)$ itself, and cannot be integrated over it. As long as we avoid such surfaces, which can for the most part be done by inspection, we will have no issue constructing our symplectic integrals. The precise mathematical criterion is that the wavefront set \footnote{The wavefront set of a distribution is a set of points in the cotangent bundle that encodes the singular directions of the distribution by a kind of local Fourier analysis. Two distributions can be multiplied when their wavefront sets do not contain additive inverses of each other \cite{wavefront}. This technology has become central to the renormalization of quantum field theories in curved spacetimes \cite{microlocal}.} of the coefficients of $\psi$ must be disjoint from the normal bundle of $H$. When the integrals are defined, currents obey Stokes' theorem as usual.
\par
In order to consistently claim that an observable $\mathcal{O}$ formally generates a distributional flow $\chi$, we would like it to satisfy the previous definition on arbitrary slices. This means we want to extend equation \eqref{conserve} to the case where one slot is a distribution. The criteria is, straightforwardly, that $\chi$ be a distributional solution to the linearized equations. This can be proven as follows, encompassing \eqref{conserve} as a special case: 
\begin{prop}
    Consider a smooth flow $\gamma$ and distributional flow $\chi$, both satisfying the linearized equations of motion. Then $d\omega(\gamma,\chi)=0$.
\end{prop}
\begin{proof}
    Varying both sides of equation \eqref{covariant}, we obtain
    \begin{equation}\label{Peierls}
        d\omega=\delta\phi^a\delta E_a.
    \end{equation}
    Inserting $\gamma$ and $\chi$, and using the fact that since they are on-shell they satisfy $\delta E_a(\gamma)=\delta E_a(\chi)=0$, we are done.
\end{proof}

\begin{defn}
    A \textit{regular} observable is one that generates a smooth flow, tangent to the prephase space. A \textit{singular} observable does not generate such a flow, but is still phase space differentiable. A \textit{singular generator} is a singular observable that "generates" a distributional flow by dualization as defined previously.
\end{defn}
The Poisson bracket of two regular observables is of course another regular observable. The Poisson bracket of a regular observable \textit{acting on} a singular observable is always defined by $\{\mathcal{S},\mathcal{R}\}=-\delta \mathcal{S}(\chi^\mathcal{R})$, where $\mathcal{R}$ is regular and $\mathcal{S}$ is singular. While it is of course possible to formally declare $\{\mathcal{R},\mathcal{S}\}=-\{\mathcal{S},\mathcal{R}\}$, this cannot be interpreted as arising from the action of $\mathcal{S}$ on $\mathcal{R}$. If $\mathcal{S}$ is a singular generator, the bracket can be interpreted as the insertion of the flows $\chi^\mathcal{R}$ and $\chi^\mathcal{S}$ into their respective slots of the integral defining the symplectic form, but we will not assume that all singular observables are generators. In cases where a distributional form of the existence properties in Assumption \ref{causal_1} apply to the source associated to $\mathcal{S}$, even if the source and inhomogeneous solutions are distributional, $\mathcal{S}$ is a singular generator with flow $X_R-X_A$. Note that even if the source is singular the flow may be smooth, as in the case of spatial smearings of the Klein-Gordon field. The bracket of two singular generators cannot generally be defined due to issues with the product of distributions. It may be meaningful in special cases where the wavefront sets of the relevant distributions are compatible, but for general pairs it will not be.
\par
If we imagine the classical field theory as the classical limit of a quantum field theory, there are natural quantum analogues of these concepts\footnote{The following discussion is for field theories without gauge invariance or a dynamical metric, but it seems that similar ideas apply even when these are included, at least in perturbation theory; see section \ref{subsec:quantum}.}. The phase space corresponds, in some sense, to the space of physically reasonable states of the QFT; more precisely, it corresponds to the set of \textit{Hadamard states}, which satisfy an operator product expansion. These are naturally regarded as more physical than other states, since they allow the definition of finite local objects such as the stress tensor. Indeed, in the classical limit the operator product expansion becomes the condition of smoothness of the field configuration \cite{OPE}, and the Schwinger-Dyson equations go to the classical equations of motion. Regular observables can be interpreted as proper (unbounded) operators that take the Hadamard states to other Hadamard states. On the other hand, singular observables should be viewed as having expectation values in Hadamard states, but not as having a well-defined operator action on them; in other words, they are \textit{quadratic forms}. Those that are singular generators further satisfy the property that the commutator with a local field operator smeared in spacetime also yields a quadratic form on the Hadamard states, so that $[\mathcal{S},\phi(x)]$ is a quadratic form-valued distribution. 
\par
This correspondence holds true in most straightforward examples. For example, a local functional of the fields and its derivatives (such as the stress tensor) is a singular generator in the classical field theory, usually after some renormalization. Smearing out such a local functional in spacetime yields a regular observable, and in simple cases, such as a free scalar field or its stress tensor, also produces a well-defined operator acting on Hadamard states. There are exceptions, where a classically regular observable fails to correspond to a proper operator; for example, if the aforementioned local functional has sufficiently high operator dimension, a purely spatial smearing will not cure all of its divergences. For instance, smearing a scalar field on a spacelike surface alone is always sufficient to guarantee that it generates a classically smooth flow, as is easily seen from the canonical structure. However, interactions can give the field a quantum anomalous dimension at short distances which makes single-time smearings fail to be an operator \cite{witten_entanglement}. While classical smoothness is not quantum mechanically protected,  the opposite should not be true. In other words, if an observable is classically singular it does not seem possible for quantum effects to render it smooth, if the classical limit is at all a sensible concept. Therefore classical regularity should be viewed as a necessary, but not sufficient, condition for quantum operatorhood on sensible states.
\par
This concludes our review and extension of covariant phase space methods and the Peierls bracket. The main new point is the systematic discussion of singular generators. Special cases have been considered previously by studying distributional kernels \cite{Pauli-Jordan, dewitt, khavkine}.

\subsection{Subsystems}\label{subsec:subsystem}
We now turn to problem of defining subsystems appropriately in the above formalism. As a warmup, we first consider a theory with nondynamical metric and no gauge symmetry. Indeed, there are no surprises here, but this analysis will provide a basic framework that will be developed further in subsequent sections. Klein-Gordon theory with polynomial self-interactions acts as a useful model for all following propositions.
\par
In such a theory a subsystem can be simply defined by fixing a partial Cauchy surface $\Sigma$, or equivalently its domain of dependence. Note that this is more than just the requirement that $\Sigma$ is spacelike; there must exist a completion of $\Sigma$ to a complete Cauchy surface, denoted $\mathfrak{C}=\Sigma \cup \overline{\Sigma}$. For example, a spacelike surface that extends to past null infinity in an asymptotically flat spacetime does \textit{not} fulfill the criteria as stated \footnote{It should be possible to remedy this by extending the covariant phase space construction to the conformally compactified spacetime, but this will require some holographic renormalization that is outside the scope of this paper.}. Similarly, although $\Sigma$ is allowed to have disconnected components or be otherwise topologically nontrivial, no two points on it may be timelike separated, or the Cauchy completion would not exist. Call the codimension-2 surface here $\Sigma$ and $\overline{\Sigma}$ intersect the \textit{entangling surface}.
\par
From the Peierls bracket point of view it is intuitively clear that the sets of observables localized in the domains of dependence of $\Sigma$ and $\overline{\Sigma}$ will form closed Poisson algebras and will commute with each other. Whenever retarded and advanced influences $X_R$ and $X_A$ can be defined for an observable located in one of the domains, it is clear that it should be determined in that domain only by the fields in that domain and by the source, which is also determined by the fields in that domain. Therefore the bracket of two observables in a domain only depends on fields in that domain. Likewise the flow $X_R-X_A$ should be extendable to a flow that vanishes in the complement by Assumption \ref{causal_2}, meaning that the two domain algebras commute with each other. However, since the Peierls bracket definition of the flow generated by an observable is not necessarily as general as the symplectic definition, it is worth going through the details to see that these facts hold from a symplectic point of view. We will also be more precise about the distinction between singular and regular observables.

\par
Define the distribution $\mathcal{S}(\Sigma)$ to consist of flows tangent to prephase space that vanish on $\overline{\Sigma}$ (we will use this terminology to mean that all derivatives vanish), and therefore have all derivatives on $\mathfrak{C}$ supported on $\Sigma$. Likewise $\mathcal{S}(\overline{\Sigma})$ consists of the flows that vanish on $\Sigma$ . It is obvious that:
\begin{prop}\label{prop:int}
    $\mathcal{S}(\Sigma)$ and $\mathcal{S}(\overline{\Sigma})$ are integrable, and therefore involutive, distributions.
\end{prop}
\begin{proof}
  We can immediately identify that $\Ss$\;foliates $\mathcal{P}$ with the path-connected submanifolds of $\mathcal{P}$ where the fields and their derivatives on $\overline{\Sigma}$ take fixed values. Similarly $\Ss$b\;generates a foliation where the fields and derivatives on $\Sigma$ take fixed values. Since these distributions form the tangent bundles of the corresponding leaves, they are involutive. This can also be easily checked by writing the Lie bracket of two vector fields in field-space components.
\end{proof}
In principle there could be two solutions that agree on $\overline{\Sigma}$ but are \textit{not} connected by a path in prephase space that preserves the restriction to $\Sigma$. In perturbation theory, where solutions are viewed as Taylor expansions of paths starting at a base point, this does not happen by tautology. Nonperturbatively it is less clear that this is true. For an ordinary scalar field, it is obvious enough that it is, since the phase space can be parametrized as a vector space $(\phi,\dot{\phi})$ consisting of the field and its time derivative on the slice, and it is easy to connect two data sets that agree on $\overline{\Sigma}$. For a general theory (especially one with gauge symmetry) there can be nonlinear constraints on the initial data that make this property less obvious, so we will need to elevate it to an assumption:
\begin{assumption}\label{path-connect}
Nonperturbatively (at least in an open neighbourhood of any point in prephase space), given two solutions that agree to arbitrary derivative order on $\overline{\Sigma}$, one can be deformed to the other while preserving the restriction to $\overline{\Sigma}$.
\end{assumption}
While this seems extremely plausible, it is harder to check this in many examples, so its status is different from that of Assumptions \ref{causal_1} and \ref{causal_2}. The purpose of this assumption is to allow us to define a \textit{phase space} $\mathcal{P}(\Sigma)$ by taking the quotient of $\mathcal{P}$ with respect to the orbits of $\Ssb$, which identifies solutions that agree on $\Sigma$. The differentiable functions on this phase space are the observables supported on $\Sigma$, which are \textit{precisely} the observables on $\mathcal{P}$ that are invariant under $\Ssb$. The phase space $\mathcal{P}(\overline{\Sigma})$ can also be defined analogously. For brevity the corresponding observables will often be rferred to as simply "in" $\Sigma$ or $\overline{\Sigma}$ respectively. In the absence of Assumption \ref{path-connect}, the quotient under $\Ssb$\;is not necessarily path-connected, and a further identification of its components is required to guarantee that solutions which agree on $\Sigma$ map to the same point. In what follows, it will be important to avoid this and identify the observables in $\Sigma$ purely in terms of their invariance under $\Ssb$. As an illustration, observables in $D(\Sigma)$ can be identified with those in $\Sigma$ (likewise for $\overline{\Sigma}$), without appealing to any properties of the equations of motion beyond the stated assumptions. This follows by differentiating the observables along paths in $\mathcal{P}$:
\begin{prop}
    The set of observables supported on $D(\Sigma)$ is exactly equal to the set of observables supported on $\Sigma$.
\end{prop}
\begin{proof}
   Consider an observable $\mathcal{O}$ in $D(\Sigma)$. Take a particular point in $\mathcal{P}$ and a tangent vector $\chi$ in $\Ssb$. The restriction of $\chi$ to $D(\overline{\Sigma})$ and the zero flow on $D(\Sigma)$ agree to all derivatives at the entangling surface. By assumption \ref{causal_2}, there is a flow on all of $M$ that is exactly equal to $\chi$ on $D(\overline{\Sigma})$ and that vanishes on $D(\Sigma)$. This flow and $\chi$ agree to all derivatives at $\mathfrak{C}$, so the difference of the two must be a null direction of $\Omega$. We have assumed that there are no such directions, but even if we were to relax this assumption, since $O$ is an observable it must be invariant under the difference. Therefore the derivative of $O$ along $\chi$ must be equal to the derivative along the constructed flow, which is zero. We conclude that $O$ is invariant under $\Ssb$, and by Assumption \ref{path-connect} is thus supported on $\Sigma$.
\end{proof}
\par
In any self-consistent classical field theory, the Poisson bracket of two regular observables is another regular observable on the same phase space. Therefore in order to to justifiably consider $\Sigma$, or equivalently $D(\Sigma)$, as defining a proper subsystem, the regular observables in $\Sigma$ must form a closed Poisson algebra. To prove this, we will first demonstrate the following:

\begin{prop}\label{complement}
$\Ss$ and $\Ssb$\;are \textit{symplectic complements} of each other; if a tangent vector $\chi$ to $\mathcal{P}$ satisfies $\Omega(\chi,\eta)=0$ for all $\eta\in\Ss$ then $\chi\in\Ssb$ and vice versa, and likewise if $\Omega(\chi,\eta)=0$ for all $\eta\in$ $\Ssb$, $\chi\in\Ss$ and vice versa.
\end{prop}

\begin{proof}
   Inserting $\chi\in\Ssb$, $\eta\in\Ss$ into $\Omega$, we find that
   \begin{equation}
\Omega(\chi,\eta)=\int_{\Sigma}\omega(\chi,\eta)+\int_{\overline{\Sigma}} \omega (\chi,\eta)=\int_{\Sigma} \cancel{\omega(0,\eta)}+\int_{\overline{\Sigma}}\cancel{\omega (\chi,0)} = 0
   \end{equation}
since $\chi$ vanishes to all derivative orders on $\Sigma$ and $\eta$ vanishes to all orders on $\overline{\Sigma}$. This establishes one direction. For the other, we will interpret the assumption that the theory has no gauge symmetry as meaning that:
\begin{claim}\label{no_gauge}
If
\begin{equation}
    \Omega(\chi,\eta)=\int_{\Sigma} \omega(\chi,\eta)=0
\end{equation}
for all $\eta\in\Ss$, then $\chi$ vanishes on all of $D(\Sigma)$. 
\end{claim}
In the next section and Appendix \ref{app:gauge}, a general result about gauge theories is proven that includes this as a special case. This implies that $\chi$ belongs to $\Ssb$ as desired.
\end{proof}
From this we can deduce the following characterization of the regular observables in $\Sigma$:
\begin{prop}\label{prop:Ss}
The set of observables generating flows belonging to $\Ss$ is \textit{exactly} the set of regular observables supported on $\Sigma$.
\end{prop}
\begin{proof}
Any observable that generates a flow which belongs to $\Ss$ for all points in $\mathcal{P}$ is invariant under any flow $\chi\in\Ssb$, and is therefore supported on $\Sigma$. Meanwhile, a regular observable supported on $\Sigma$ is invariant under all flows in $\Ssb$, so the flow that it generates must symplectically annihilate any flow in $\Ssb$. By Proposition \ref{complement}, this flow belongs to $\Ss$ for all points in $\mathcal{P}$.
\end{proof}
As a consequence we have two corollaries:
\begin{cor}
    The regular observables supported on $\Sigma$ form a closed Poisson algebra.
\end{cor}
\begin{proof}
    The bracket of two observables that each generate flows in $\Ss$ is a regular observable that generates the Lie bracket of those flows. By Proposition \ref{prop:int}, $\Ss$ is involutive, so this flow also belongs to $\Ss$. The converse direction of Proposition \ref{prop:Ss} implies that the bracket is an observable supported on $\Sigma$.
\end{proof}
\begin{cor}
    The regular observables supported on $\Sigma$ Poisson commute with all observables localized in $\overline{\Sigma}$.
\end{cor}
\begin{proof}
    This is a trivial consequence of the definition of $\Ss$.
\end{proof}
\par
The key step in this chain of arguments that made use of the causality assumptions is the characterization of theories with no gauge symmetry given in Claim \ref{no_gauge}. The regular observables in $\Sigma$ are therefore consistent with the proposal $\Sigma$ defines a subsystem. Recalling that these are supposed to correspond to sensible operators in the quantum field theory, we could be satisfied. However we have not said anything about the \textit{completeness} of the regular observables; are their values sufficient to determine the state in $\mathcal{P}(\Sigma)$? At any point in $\mathcal{P}$, there is a linearized observable that generates any flow in $\Ss$. Consider two arbitrary smooth perturbations $\zeta_1$ and $\zeta_2$, and suppose that $\Omega(\zeta_1-\zeta_2,\eta)=0$ for all $\eta\in\Ss$. Again, this implies that $\zeta_1$ and $\zeta_2$ are equal on $D(\Sigma)$, so they induce the same flow on $\mathcal{P}(\Sigma)$. Therefore the flow that any $\zeta$ induces on $\mathcal{P}(\Sigma)$ is determined entirely by its insertion into the linearized regular observables supported on $\Sigma$. 
\par
However, it is not obvious that this remains true in the full nonlinear theory. The aforementioned linearized observables at any point of $\mathcal{P}$ are really the 1-forms $\Omega(\cdot,\eta)$, which form a subbundle of $T^*\mathcal{P}$. To carry out the argument above for the nonlinear regular observables, it is necessary to find a basis of this bundle consisting of exact forms In other words:
\begin{assumption}\label{complete}
At each point in $\mathcal{P}$, every flow in $\Ss$ is generated by some regular observable supported on $\Sigma$.
\end{assumption}
Note that if $\mathcal{P}$ were finite dimensional, it would be possible to \textit{deduce} this from Proposition \ref{complement} using the fact that all 1-forms can be obtained by inserting some flow into $\Omega$, and therefore this bundle is exactly the set of 1-forms that annihilates $\Ssb$. Since $\Ssb$\;is an integrable distribution, the codistribution that annihilate it admits an exact basis. The flaw of this argument in the Frechet setting is that some of these annihilating 1-forms do not generate a smooth flow. The bundle of \textit{regular} 1-forms annihilating $\Ssb$\;is generated by inserting flows belonging to $\Ss$\;into $\Omega$, but this is not the entire annihilator of $\Ssb$ and is not guaranteed to be integrable. The best we can do is to note that if a functional generates a smooth flow at some point in $\mathcal{P}$, as the solution is smoothly deformed it will likely continue to generate a smooth flow in some open neighbourhood around that point. I expect that this can be made precise enough to to prove Assumption \ref{complete} in perturbation theory, but I will not pursue this further. This is the least important of the assumptions we've introduced, and is not used directly to prove any other results. Some further discussion on this theme can be found in section \ref{subsec:factor}.
\par
We may also ask whether the Poisson bracket of observables supported on $\Sigma$ can be computed via a symplectic form constructed locally out of the fields and their variations on $\Sigma$, which induces a symplectic form on $\mathcal{P}(\Sigma)$. This is the approach to studying subregion phase spaces taken in much of the literature \cite{subsystem,local,fixed_area,reference,frames}. The natural choice is to define $\Omega_\Sigma=\int_{\Sigma} \omega$. Unfortunately, this is subject to some ambiguity because $\omega$ can be shifted by local exact forms, as discussed in section \ref{subsec:prelim}. The resolution of this ambiguity for the total symplectic form $\Omega$ \cite{harlow, Compere_Marolf_corner} based on an appropriate choice of boundary conditions does \textit{not} resolve this issue for $\Omega_\Sigma$ if $\Sigma$ has any nontrivial boundary components in the interior of $\mathcal{M}$. A term $d\beta$ added to $\omega$ for which $\beta$ dies off sufficiently fast towards the global spacetime boundary will not affect $\Omega$, but can contribute a nontrivial term $\int _{\partial\Sigma} \beta$ to $\Omega_\Sigma$. Kirklin has argued \cite{unambiguous} that this ambiguity disappears in a different formalism where the symplectic current is integrated around a closed surface enclosing $\Sigma$ instead, and on-shell flows are replaced with flows obeying certain retarded boundary conditions. We will see that the ambiguity in fact goes away even in the original covariant phase space formalism.
\par
First note that:
\begin{prop} 
The Poisson brackets of regular observables supported on $\Sigma$, computed according to $\Omega_{\Sigma}$, agrees with that computed according to $\Omega$. This holds true regardless of how the symplectic current is fixed in the interior of $\mathcal{M}$.
\end{prop}
\begin{proof}
    Suppose that such an observable $O$ generates a flow $\eta$ according to $\Omega$. Since $\eta\in$ $\Ss$,
    \begin{equation}
        \delta O = \int \omega(\cdot,\eta) = \int_{\Sigma} \omega(\cdot,\eta) + \int_{\overline{\Sigma}} \cancel{\omega(\cdot,\eta)}= \Omega_\Sigma(\cdot, \eta).
    \end{equation}
    Despite this, $O$ does not uniquely generate $\eta$ on $\mathcal{P}$ via $\Omega_\Sigma$ because there exist other flows that have this property. However, any such flow differs from $\eta$ by a vector $\chi$ such that $-I_\chi \Omega_\Sigma=0$, which implies that $\chi\in$ $\Ssb$ and vanishes in $D(\Sigma)$. On the phase space $\mathcal{P}(\Sigma)$, there is thus a unique flow generated by $O$. Under the action of this flow, the variation of observable in $\Sigma$ is the same as its variation under $\eta$. The ambiguity in the symplectic current does not affect this argument because if we were to shift $\omega$ by a locally constructed term $d\beta$, the change in $\Omega_\Sigma$ is $\int_{\partial\Sigma}\beta$, and since $\eta$ goes to zero at all derivative orders near $\partial \Sigma$, $\int_{\partial\Sigma}\beta(\cdot,\eta)=0$.
\end{proof}
The key assumption in this proof was that \textit{one} of the observables was regular, in order to define the flow that it generates. There is no problem if the other observable is singular, as long as the Poisson bracket is taken in the appropriate direction i.e. ${\cdot,O}$. In fact with a slight modification this argument can be extended to when the $O$ is a \textit{singular generator} and the other is regular. The key observation is that even if a flow $\eta$ is distributional, in order for $\Omega_{\overline{\Sigma}}(\chi,\eta)=0$ for all flows $\chi$ in $\Ssb$, $\eta$ must restrict to zero on the interior of $D(\Sigma)$. Again, this follows from a more general result in the next section. The proof of Proposition \ref{prop:Ss} extends immediately to show that the singular generators supported on $\Sigma$ are exactly those that generate flows vanishing in $D(\Sigma)$. This implies the following:
\begin{prop}
  Given a regular observable $R$ and a singular generator $O$ supported on $\Sigma$, the Poisson bracket $\{R,O\}$ is identical whether computed with $\Omega$ or $\Omega_{\Sigma+\Delta}$, where $\Sigma+\Delta$ is an arbitrarily small enlargement of $\Sigma$ within the Cauchy slice $\mathfrak{C}$.
\end{prop}
\begin{proof}
    The bracket is given by
    \begin{equation}
        \{R,O\}=\Omega(\eta_R,\eta_O)=\int_{\Sigma+\Delta} \omega(\eta_R,\eta_O)+\int_{\overline{\Sigma+\Delta}} \cancel{\omega(\eta_R,\eta_O)}=\Omega_{\Sigma+\Delta}(\eta_R,\eta_O)
    \end{equation}
as desired, because both $\eta_R$ and $\eta_O$ vanish on the interior of $\overline{\Sigma}$. The purpose of the enlargement of $\Sigma$ is to enclose within the subregion possible distributional singularities that may appear at the boundary of $\Sigma$. These integrals are in terms of currents as discussed in section \ref{subsec:prelim}, so the integral over ${\overline{\Sigma+\Delta}}$ vanishes because $\omega(\eta_R,\eta_O)$ is the zero current there.
\par
If the symplectic current were shifted by a term $d\beta$, the result would be to shift $\Omega_{\Sigma+\Delta}(\eta_R,\eta_O)$ by $\int_{\Sigma+\Delta}d\beta(\eta_R,\eta_O)$. By the Stokes' theorem for currents, this still integrates out to $\int _{\partial(\Sigma+\Delta)} C(\eta_R,\eta_O)=0$. 
\end{proof}
\begin{figure}
    \centering
    \includegraphics[scale=.2]{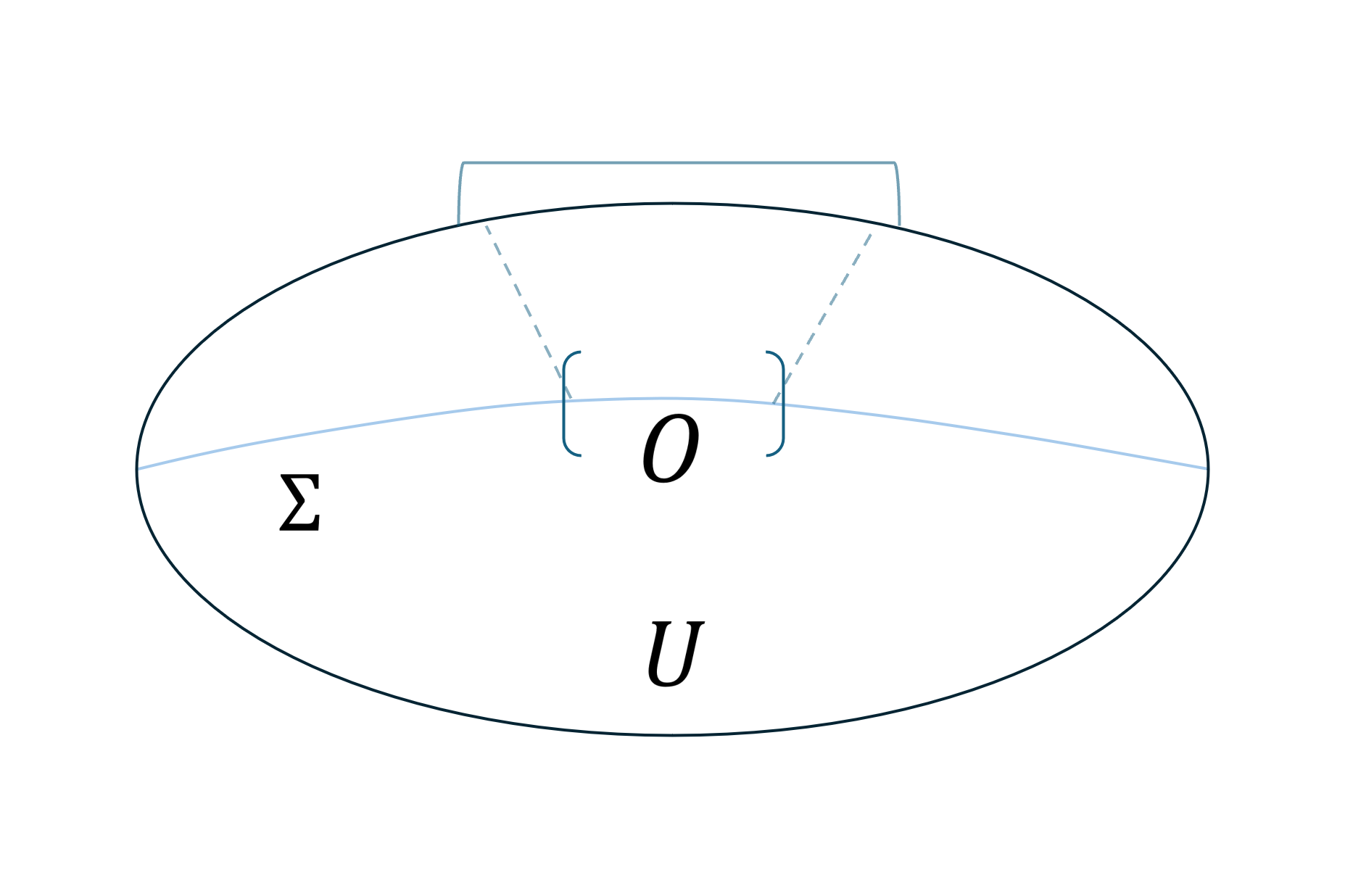}
    \caption{Kirklin's prescription for resolving ambiguities for singular flows. $O$ is taken to generate a flow that is supported within its future light cone (indicated by the dotted lines), and the insertion into $\omega$ is subsequently integrated over $\partial\Sigma$. Note that this is equivalent to integrating over the part of $\partial U$ to the future of $\Sigma$, which is equivalent to simply integrating over $\Sigma$.}
    \label{fig:unambiguous}
\end{figure}
Technically the enlargement of $\Sigma$ was not strictly necessary in this argument because the flow generated by $R$ went to zero near $\partial_\Sigma$. However, one could consider the case where $R$
 is also a singular generator such that the Poisson bracket with $O$ is well-defined. As noted in section \ref{subsec:prelim}, this requires that the distributional singularities of the flows be compatible in a manner encoded by the wavefront set. Supposing that this is the case, the enlargement to $\Sigma+\Delta$ computes the Poisson bracket unambiguously in terms of the local symplectic form. This is the same mechanism behind Kirklin's description of the subregion bracket, illustrated in figure \ref{fig:unambiguous}. In that formalism one integrates the symplectic current on the boundary of an open set $U$ enclosing $\Sigma$, but modifies the flow generated by an observable supported on $\Sigma$ to be zero to the past of $\mathfrak{C}$ and equal to the original flow to its future (note that this is no longer an on-shell vector tangent to $\mathcal{P})$. One can see by inspection that the insertion of two such flows into $\omega$ yields a current which vanishes outside of a part of $\partial U$ to the future of $\Sigma$. The on-shell conservation of $\omega$ implies that this is equal to the insertion of the original on-shell flows into $\omega$, integrated over $\mathfrak{C}$. In both cases the ambiguity due to $\beta$ manifestly vanishes if one considers a region slightly larger than the support of the flows, and applies Stokes' theorem to integrate $d\beta$ out to zero.
 \par
 The next three sections will primarily focus on studying regular observables and their Poisson brackets in gauge and gravitational theories. In \ref{subsec:surface}, we will return to the singular generators and find that they play an important role in the study of surface charges.

\section{Gauge Theories}
\subsection{Characterizing Gauge Symmetry}\label{subsec:gauge_character}
To extend the analysis of the previous section to cover gauge theories, we will need to deliver the promised extension and proof of claim \ref{no_gauge}. This will require a deeper characterization of the local structure of gauge symmetries.  Noether \cite{Noether} originally defined as infinitesimal transformations of the fields that are locally constructed out of the fields, their derivatives, and some free functions $\lambda^j (x)$, which transform the Lagrangian by a likewise locally constructed exact form. In covariant phase space language, the symmetry $X_\lambda$ satisfies $\delta L(X_\lambda)=d\alpha_\lambda$. Noether's second theorem states that there is a one-to-one correspondence between the \textit{independent} parameters $\lambda^j$ to the nontrivial relations among the equations of motion that hold off-shell, $\mathcal{N}^{ja} E_a=0$, where $\mathcal{N}^{j}$ are differential operators \cite{antifield_lectures}. A derivation in modern language can be found in \cite{Noether_Ward}. The Noether identities for any theory can be determined in closed form by an algebraic computation; for instance, if the equations of motion are written in terms of curvature forms associated with vector bundles, all Noether identities follow from the Bianchi identities. For Yang-Mills theory and GR, this implies that the standard gauge symmetries (internal Lie rotations and diffeomorphisms respectively) are the \textit{only} Noetherian gauge transformations.
\par
These flows $X_\lambda$, and indeed any symmetries of the Lagrangian, are always tangent to the prephase space \cite{harlow}. While they are not defined with respect to the symplectic form, it is easy to see that $X_\lambda$ do constitute null directions of the symplectic form, since the parameters $\lambda$ can be chosen arbitrarily in the neighbourhood of the Cauchy slice $\mathfrak{C}$ and switched off to the  future to make the flow vanish near another Cauchy slice. By conservation of the total symplectic form, this guarantees that $I_{X_\lambda} \Omega=0$. Therefore Noether's definition is compatible with the symplectic notion of a gauge transformation. We would like to say something in the opposite direction, connecting the symplectic gauge transformations to the Noetherian symmetries. It is helpful to define the following class of flows:
\begin{defn}\label{def:gauge}
    At a particular point in $\mathcal{P}$, a flow $X$ tangent to prephase space is \textit{locally gauge} in an open set $U$ if
    \begin{equation}\label{local_gauge_def}
        -I_X\omega=C+dq
    \end{equation}
  within $U$, where C is a spacetime $D-1$ form, field-space 1-form that is a linear combination of linearized equations of motion $\delta E_a$ and $q$ is some locally constructed spacetime $D-2$ form, field space 1-form.
\end{defn}
This condition means that $-I_X\omega$ integrates out to a surface term on any partial Cauchy surface localized within $U$. In particular, a flow that is locally gauge everywhere and has nontrivial $\int_{\partial\mathfrak{C}} q$ is a \textit{surface symmetry} of the theory, which is a generalization of the concept of asymptotic symmetry that also applies to finite boundaries. The ADM Hamiltonian for general relativity \cite{surface_integrals} or the smeared electric flux for electromagnetism are good examples of surface symmetry generators, but many more have been found throughout the long history of the field \cite{brown_henneaux, harlow,subsystem, causal_diamond, strominger}.  
\par
Noetherian and symplectic gauge transformations are now connected by the following pair of results:
\begin{prop}\label{prop:local_gauge}
Given a solution in $\mathcal{P}$ and a partial Cauchy surface $\Sigma$, any flow $X$ such that 
\begin{equation}\label{gauge_annihilation}
    \Omega_\Sigma(X,\eta)=0
\end{equation}
for all $\eta\in$ $\Ss$ must be locally gauge within $D(\Sigma)$.
\end{prop}
\begin{prop}\label{prop:local_gauge_noether}
    The locally gauge flows at a point of $\mathcal{P}$ are Noetherian gauge symmetries of the \textit{linearized} theory about that point. This applies within any open region.
\end{prop}
The proofs of these statements are fairly technical, and are relegated to Appendix \ref{app:gauge}. For now, we will focus on their application.  Since a symplectic gauge symmetry satisfies equation \eqref{gauge_annihilation} for the entire Cauchy slice $\mathfrak{C}$, it must be locally gauge everywhere in $\mathcal{M}$. There are in general flows that are locally gauge everywhere but are \textit{not} null directions of $\Omega$; these are the asymptotic symmetries of the theory. Proposition \ref{prop:local_gauge_noether} then tells us that all locally gauge flows can be determined by applying Noether's analysis to the linearized equations of motion. This amounts to finding the Noether identities among these equations, which is again a purely algebraic exercise.
\par
There are two subtleties in Propositions \ref{prop:local_gauge} and \ref{prop:local_gauge_noether} that require care. First, consider again a theory with no gauge symmetry and suppose that we consider a surface $\Sigma'$ that is homologous to $\Sigma$ but which extends outside of $D(\Sigma)$, and is therefore not strictly spacelike (illustrated in figure \ref{fig:causal_example}). Take a smooth flow $Y\in$ $\Ssb$, which vanishes in $D(\Sigma)$ but not on $\overline{\Sigma}$. This flow will propagate within the light cones of $\overline{\Sigma}$, which intersect $\Sigma'$. Therefore $Y$ will in general not vanish on $\Sigma'$. However, since $Y$ vanishes on $\Sigma$, $I_Y\Omega_\Sigma=0$, and by the conservation of $\omega$,  $I_Y\Omega_{\Sigma'}=0$ as well. This demonstrates that Proposition \ref{prop:local_gauge} must fail in the absence of certain causality properties for the subsystem. The proof presented in Appendix \ref{app:gauge} makes essential use of the causality assumptions \ref{causal_1} and \ref{causal_2} as well as the character of $\Sigma$ as a partial Cauchy surface, to characterize the local nature of gauge symmetries that is to my knowledge essentially new.

\begin{figure}
    \centering
    \includegraphics[scale=.2]{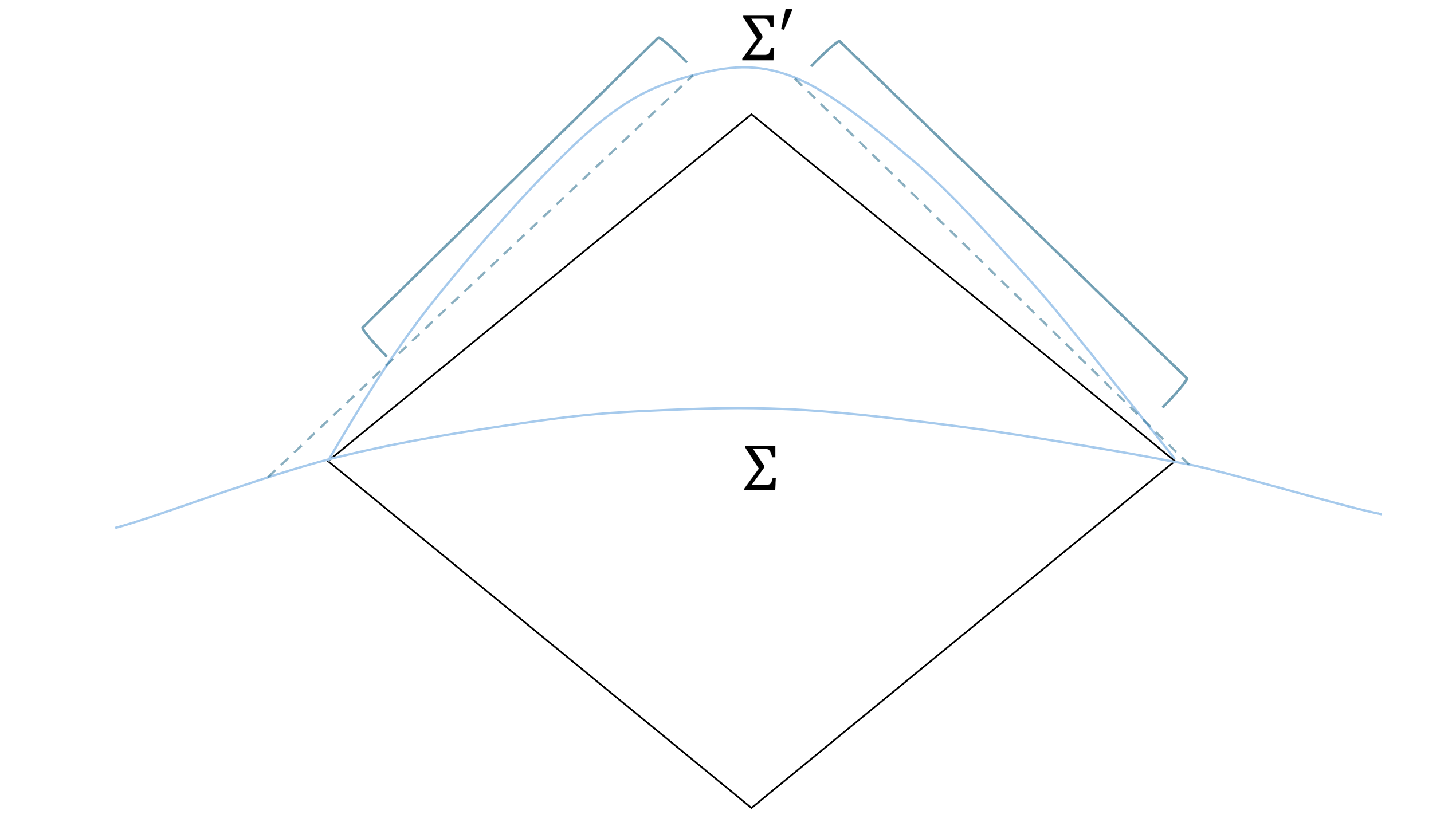}
    \caption{A flow in $\Ssb$ can propagate to $\Sigma'$. The support of this flow on $\Sigma'$ is indicated by the square brackets. }
    \label{fig:causal_example}
\end{figure}

\par
The other subtlety is that Proposition \ref{prop:local_gauge_noether} relates the local gauge transformations to the \textit{linearized} Noetherian gauge symmetries, not those of the complete interacting action. This correspondence has been previously argued in a paper by Vitagliano \cite{Vitagliano}, although this did not include Proposition \ref{prop:local_gauge}. The linearized gauge symmetries always include the nonlinear ones, but there may also be extra ones. This happens when, at a particular reference solution in $\mathcal{P}$, the linearized equations $\delta E_a$ obey additional Noether identities that do \textit{not} hold off-shell for the equations $E_a$. This means that there are fewer independent linearized equations than there are nonlinear equations, which implies that some flows $X$ satisfying $\delta E_a(X)$ cannot be extended to paths tangent to phase space. This is well-known as the phenomenon of \textit{linearization instability}, and occurs when some of the gauge symmetries with nonzero $\lambda$ leave the reference solution fixed (locally within the open set under consideration in \ref{prop:local_gauge_noether}) \cite{fischer_marsden,moncrief_1,moncrief_2}. In GR this corresponds to the existence of local Killing vectors. There is an extensive literature on this topic; a convenient recent source that studies these instabilities in covariant language is \cite{Kirklin_instability}.
\par
At a solution which has such a local instability, a flow can be locally gauge without corresponding to a Noetherian gauge transformation of the full theory. In pure GR, for instance, a flow could be locally gauge within an open set without being induced by a diffeomorphism. This could give rise to a pathology in some subsequent arguments; however, this is only expected to happen at a sparse set. We will be interested in observables that generate flows that are locally gauge in some region, and the flows will be Noetherian symmetries at all solutions that are linearization stable. We will mostly restrict to such stable solutions, which are dense in the solution space; any closure or commutativity properties of the algebra will pass through in the limit. 
\par
As a further consequence, we can establish that gauge-invariant observables with a well-defined temporally bounded trace are indeed compatible with the Peierls bracket:
\begin{cor}
    For an observable $\mathcal{O}$ such that $\frac{\delta\mathcal{O}}{\delta\phi}=J_a$, gauge invariance is equivalent to $\Delta(J_a)=0$ for all of the linear Noether identities. 
\end{cor}
\begin{proof}
    By Proposition \ref{prop:local_gauge_noether} applied to a complete Cauchy slice, symplectic gauge transformations are precisely the Noetherian gauge symmetries everywhere in spacetime. Therefore gauge invariance implies that $\mathcal{O}$ is invariant under the Noetherian flows. This implies that $\int J_a X_\lambda^a$ is zero for all compactly supported Noetherian symmetries. Integrating by parts to move derivatives off of the gauge parameters $\lambda$, this can be rewritten as $\int\Delta_\lambda(J_a)\lambda=0$ with an implicit sum over the gauge parameters (this relationship between the Noether identities and gauge symmetries is discussed in Appendix \ref{app:gauge}). Therefore the Noether identities hold for source $J_a$. The converse, that the identities imply gauge invariance, was argued in section \ref{subsec:singular} when the Peierls bracket was introduced.
\end{proof}

\subsection{Subsystems in Gauge Theory}\label{subsec:gauge_subsystem}
The point where the reasoning of section \ref{subsec:subsystem} fails for gauge theories is Proposition \ref{complement}. The symplectic complement of $\Ss$ is larger than $\Ssb$; it also includes all of the gauge transformations. We can define augmented distributions $\Ssp= \Ss+\mathcal{G}$ and $\Ssbp=\Ssb+\mathcal{G}$, where $\mathcal{G}$ is the distribution on $\mathcal{P}$ consisting of symplectic gauge transformations. 
\begin{prop}\label{complement_gauge}
    $\Ssp$ and $\Ssbp$ are symplectic complements of each other. 
\end{prop}
\begin{proof}
    Consider flows in $\Ssp$ and $\Ssbp$ represented by $\eta+\zeta_1$ and $\chi+\zeta_2$ respectively, where the $\eta$ belongs to $\Ss$, $\chi$ belongs to $\Ssb$, and the $\zeta$'s are pure gauge. The symplectic form contracted with these flows can be written as:
    \begin{equation}\label{}
\Omega(\eta+\zeta_1,\chi+\zeta_2)=\Omega(\eta,\chi)=\Omega_\Sigma(\eta,0)+\Omega_{\overline{\Sigma}}(0,\chi)=0,
    \end{equation}
which verifies that the flows in $\Ssbp$ symplectically annihilate those in $\Ssp$. To see the converse, consider an arbitrary flow $\gamma$ such that $\Omega(\gamma,\eta+\zeta_1)=0$. It follows that $\Omega(\gamma,\eta_1)=0$ as well, so by Proposition \ref{prop:local_gauge}, $\gamma$ must be locally gauge within $D(\Sigma)$. The restriction of $\gamma$ to this domain of dependence can be described by some gauge parameter functions $\lambda^a$, which can be extended arbitrarily (up to global boundary conditions) throughout the rest of $\mathcal{M}$ to define a flow $\zeta_2$ that is locally gauge everywhere. We can then express $\gamma$ as $\chi+\zeta_2$, where $\chi\in\Ssb$. By choosing the parameters $\lambda^a$ to fall off sufficiently quickly away from $\partial\Sigma$ within $D(\overline{\Sigma})$, 
$q_{\zeta_2}$ can be localized to $\partial M\cap \Sigma$. If this yields a nontrivial linearized observable, there is some $\chi\in\Ss$ that alters its value, contradicting the assumption that $\Omega(\gamma,\chi)=0$. This means that $\zeta_2$ must be pure gauge, which exhibits $\gamma$ as a member of $\Ssbp$ as desired. The other direction follows from the same arguments \textit{mutatis mutandis}.
\end{proof}
\begin{cor}\label{causal_complement_gauge}
    The flows that belong to $\Ssp$ are precisely those that restrict to some gauge transformation $\eta$ within $D(\bar{\Sigma})$. Similarly those belonging to $\Ssbp$ are exactly those that restrict to a gauge transformation within $D(\Sigma)$. 
\end{cor}

We can also establish that:
\begin{prop}\label{inv_gauge}
    $\Ssp$, and likewise $\Ssbp$, are involutive distributions.
\end{prop}
\begin{proof}
    Consider two sections $\eta_1+\zeta_1$ and $\eta_2+\zeta_2$ of $\Ssp$. Their Lie bracket can be expanded as 
    \begin{equation}\label{bracket}
        [\eta_1+\zeta_1,\eta_2+\zeta_2]=[\eta_1,\eta_2]+[\eta_1,\zeta_2]+[\zeta_1,\eta_2]+[\zeta_1,\zeta_2].
    \end{equation}
    The first term in this expansion is the bracket of two flows in $\Ss$, which is also in $\Ss$ by Proposition \ref{prop:int}. The last term is the bracket of two gauge transformations, which is also gauge. Focus on the term $[\eta_1,\zeta_2]$ and introduce a component representation $\eta_1=\eta^a\frac{\partial}{\partial\phi^a}$, as well as a condensed notation $\zeta_2=Z(\lambda^j)$ for the gauge transformation in terms of the parameters $\lambda^j$. The parameters $\eta^a$ and $\lambda^j$, as well as the map $Z$, are field-dependent. The latter is local in the fields, parameters, and their derivatives. We have
    \begin{equation}
        [\eta_1,\zeta_2]=\eta^a\frac{\partial Z}{\partial \lambda^j}\frac{\partial \lambda^j}{\partial\phi^a}+\eta^a\frac{\partial}{\partial\phi^a}Z(\lambda^j)-\delta_{Z(\lambda^j)} \eta^a \frac{\partial}{\partial \phi^a}.
    \end{equation}    The first term differentiates $\lambda^j$ under some flow at fixed $Z$, and therefore returns a gauge transformation \footnote{Actually it is only required to return a locally gauge flow. However if the fall-off behaviour required for $\lambda^j$ to induce a true gauge transformation is solution-independent, as it is for standard boundary conditions for GR and Yang-Mills theory, then this term is properly gauge.}. The second differentiates $Z$ with respect to $\eta$ at fixed $\lambda^j$. The resulting vector is not necessarily gauge, but the locality of $Z$ implies that its spacetime support is contained in that of $\eta^a$, which in turn implies that it vanishes on $D(\overline{\Sigma})$. The final term differentiates the flow $\eta$ itself with respect to $\zeta_2$, but since $\eta$ vanishes on $D(\Sigma)$ for all solutions, so does this term. As a result $[\eta_1,\zeta_2]$ is equal to the sum of a gauge transformation and a configuration space vector that vanishes in $D(\overline{\Sigma})$. Since the Lie bracket of two on-shell flows is also on-shell, that vector must belong to $\Ss$, and the whole term belongs to $\Ssp$. A similar argument shows that $[\zeta_1,\eta_2]\in\Ssp$, so the total Lie bracket is a section of $\Ssp$ as desired.
\end{proof}
Unlike in section \ref{subsec:subsystem}, it is not always true that $\Ssp$ and $\Ssbp$ are integrable. If they \textit{were}, we could immediately construct the phase spaces $\mathcal{P}(\Sigma)$ and $\mathcal{P}(\overline{\Sigma})$ by taking the quotient of $\mathcal{P}$ along the orbits of $\Ssp$ and $\Ssbp$ respectively. The differentiable functions on each phase space would be gauge invariant by construction, since the distributions include the gauge transformations, while also being localized in their respective regions. However for generally covariant theories, it is well known that there are \textit{no} such observables, because the action of finite diffeomorphisms can always move a localized functional into a region where it originally had no support, contradicting its presumed gauge invariance. The failure of the Frobenius theorem in Frechet spaces, which originally seemed a pure technicality, is now clearly necessary for logical consistency!
\par
Another way to look at what happens here is that the diffeomorphisms do not act as a Lie group on the fields in any spacetime subregion, including for instance $D(\Sigma)$. Although the effect of an \textit{infinitesimal} diffeomorphism is determined only by the fields in the region, integrating some parameter dependent diffeomorphism can bring field data from the exterior to the interior. Since the configuration in the interior does not constrain the exterior beyond derivatives along the common boundary, it is impossible for a functional to be simultaneously diffeomorphism invariant \textit{and} independent of the exterior. This leads us to a conjecture, which is in fact true:
\begin{prop}\label{gauge_int}
If the gauge transformations act as a Lie group on the fields (and their derivatives) on $\Sigma$, $\Ssbp$ is integrable. The same is true for $\overline{\Sigma}$, \textit{mutatis mutandis}.
\end{prop}
\begin{proof}
We can define the foliation by identifying solutions that belong to the same leaf. This is carried out in two steps; first, identify solutions that agree on $\Sigma$. The gauge transformations have a well-defined Lie group action on this space of equivalence classes, so we can mod out the orbits to obtain an additional identification. The combined identification defines the leaves of the desired foliation of $\mathcal{P}$; by Assumption \ref{path-connect} and the connectedness of the gauge group, both steps are identified by smooth paths, so the leaves are path-connected. The arguments also hold symmetrically for $\overline{\Sigma}$.
\end{proof}
 This property applies to ordinary nongravitational gauge theories, such as Yang-Mills. Assumption \ref{path-connect} tells us that any two solutions such that, after making some finite gauge transformation, agree on $\overline{\Sigma}$ to all derivative orders, are on the same leaf of the foliation. Therefore the construction of the phase spaces $\mathcal{P}(\Sigma)$ and $\mathcal{P}(\overline{\Sigma})$ go through as before, and the functions on them are the gauge invariant observables left invariant by $\Ssbp$ and $\Ssp$ respectively.
\par
While a regular observable does not generate a fixed flow on $\mathcal{P}$ due to the gauge redundancy, it does so up to the addition of an arbitrary field-dependent infinitesimal gauge transformation (a section of $\mathcal{G}$). We can prove the following replacement for Proposition \ref{prop:Ss}:
\begin{prop}\label{prop:Ssp}
    The set of observables generating flows belonging to $\Ssp$ is precisely the set of regular observables supported on $\Sigma$. 
\end{prop}
\begin{proof}
    The proof is essentially the same as that of the earlier version. The flow $\eta$ generated by $O$, if $O$ is supported in $\Sigma$, must symplectically annihilate any flow $\chi\in\Ssbp$ since $O$ is invariant under $\Ssbp$.  By Proposition $\ref{complement_gauge}$ $\eta\in\Ssp$. Likewise if $O$ generates a flow in $\Ssp$, $\delta O(\chi)=0$ for all $\chi\in\Ssbp$, implying that $O$ is supported on $\Sigma$. 
\end{proof}
This means that these observables generate \textit{field-dependent gauge transformations} on the causal complement. Of course, since a flow in $\Ssp$ can be divided into a gauge transformation and flow in $\Ss$, one might as well just take the flow to vanish in $D(\overline{\Sigma}$. However, this concept will come into play in an important way in the next section, where it is no longer possible to select the flow on the complement to vanish.
\par
\begin{cor}
    The set of regular observables localized in $\Sigma$ form a Poisson algebra.
\end{cor}
\begin{proof}
    By Proposition \ref{inv_gauge}, the bracket of two such observables generates a flow in $\Ssp$, and thus is also localized in $\Sigma$.
\end{proof}

\begin{cor}
    The regular observables localized in $\Sigma$ Poisson commute with all regular observables localized in $\overline{\Sigma}$.
\end{cor}
\begin{proof}
    The flow generated by such a regular observable is equal to some gauge transformation on $\overline{\Sigma}$, so preserves the value of any gauge invariant observable within that region.
\end{proof}

There may be no symplectic form defined on $\Sigma$ alone that produces this bracket, but it is still the correct bracket for the subsystem. It can be computed locally for regular observables by computing the flows that they generate according to $\Omega_\Sigma$ and stripping off a locally gauge contribution such that the leftover flows vanish to all orders at the edge.

\section{Dressed Subsystems and Gauge Fixing}
\subsection{Gravitational Dressing}\label{subsec:dress}
Now we come to the main objective of this study; to extend the definition of subsystems to gravitational theories, that are both generally covariant and have a variable causal structure. As noted before, a fixed subregion with respect to a background manifold structure contains no observables. In order to obtain nontrivial observables, the location of the subsystem must depend on the state; in other words, it must be determined \textit{in relation} to some dynamical degrees of freedom. Donnelly and Freidel \cite{subsystem} introduced a tool for describing subregions with fluctuating boundaries that has become widely used in the literature \cite{local}. The idea is to introduce additional variables consisting of a map $X$ that embeds a reference manifold $R$ into the spacetime $M$. By pulling back the fields along this map and explicitly accounting for the embedding map in the symplectic structure, they obtained an extended phase space on $R$ that they identified with a subsystem in the abstract. However, since the embedding maps are additional variables that do not appear in the original action, the gauge invariant functions on this phase space are not expressible as observables constructed purely in terms of the original dynamical fields on $M$. Thus that construction is not sufficient to describe a subsystem defined in relation to the physical degrees of freedom, until further information about the fields on $M$ is incorporated into $X$. An extensive formalism for constructing relational observables was developed in \cite{reference_frames} (see also \cite{reference,frames}), which makes use of "dynamical reference frames". These are essentially an incarnation of the embedding maps $X$ that are constrained with respect to the dynamical fields on $M$. This method, rather than the extended phase space approach, can be used to identify subregions that contain genuine physical observables. However it is not at all clear based on that work whether these observables form a closed Poisson algebra. We will find that with a single additional, highly natural condition on the frames, they do, and that the regions defined this way constitute proper subsystems.
\par
Rather than invoke the full mathematical sophistication of \cite{reference_frames}, we will use a simpler description that is equivalent for our purposes. For each solution in $\mathcal{P}$, a dynamical reference frame is a smoothly connected class of allowed maps $X$ that embed a $D-1$-dimensional manifold as a partial Cauchy surface $\Sigma$. The constraints on $X$ are constructed purely out of the dynamical fields on $M$, and may only be defined inside an open region of prephase space\footnote{For example, we could consider a class of solutions containing a single worldtube of high energy density, representing a planet, for which an average worldline can be defined.  Given such a solution, neighbouring solutions will also have this property. Then in all of those solutions, the maps $X$ could be constrained so that $\partial\Sigma$ is the edge of the timelike envelope of this worldline. More examples are discussed in section \ref{sec:examples}.}.
The constraints are required to satisfy the following properties:
\begin{enumerate}[I.]
    \item They fully fix the region; for a given solution, all allowed $X$ induce the same $\partial\Sigma$.
    \item They are covariant; if a finite diffeomorphism $\varphi$ is applied to a solution for which $X$ is an allowed map, $\varphi\circ X$ is an allowed map for the transformed solution. 
    \item They are internal to the region; a solution is perturbed in a way such that the fields and their derivatives are left unchanged on $D(\Sigma)$, $X$ is also an allowed map for the new solution.
\end{enumerate} 

The first two properties are encoded in \cite{reference_frames}, but the third is new. To interpret it, consider an observer with access to an experimental apparatus living on $\Sigma$, or equivalently in its domain of dependence. If the observer has knowledge of the constraints on $X$, then they can determine the allowed values of $X$ by taking measurements only within the region. Most importantly, they have a protocol for determining the edge of the region without exiting it. This means that a measurement system can be consistently confined within the region. This is a reasonable demand to make of an observable subsystem in a gravitational theory. We call such a region an \textit{internally dressed subsystem}. 
\par
To verify whether or not these can actually be regarded as proper subsystems, we need to identify the observables contained therein and study their brackets. We should expect the brackets to close from the Peierls bracket perspective by a familiar argument; for any solution, the source obtained by functionally differentiating an observable localized in an internally dressed subsystem will also be localized in that region. Therefore the retarded and advanced influences should depend only on the fields within the subsystem. If the system is not an internally dressed, then its observables can implicitly depend on exterior fields through the the reference frame $X$, so the source will generally extend outside of the region. 
\par
Nevertheless, we need to do further work to establish that this closure is true for general symplectically regular observables. To make this process more convenient, note that Property II implies that by applying a diffeomorphism to any solution, it is possible to ensure that a particular fixed map is allowed. We will additionally require that $X$ satisfies strong enough boundary conditions near $\partial M$ that one can always do this by making a pure gauge transformation\footnote{Otherwise, the embedding of $\Sigma$ near the global boundary depends on the solution, and cannot be transformed away with an ordinary gauge transformation. For instance, in an asymptotically AdS setting one could consider the causal wedge of a boundary region that changes size depending on certain field values. This is similar to taking a field-dependent subregion in a nongravitational theory, which certainly does not define a subsystem in any useful sense. I do not expect that one will find a closed Poisson algebra in these cases.}, in which case imposing that the selected $X$ is allowed amounts to a choice of \textit{gauge fixing}. Property I implies that this gauge fixing renders the location of $\partial\Sigma$ constant, so that residual diffeomorphisms cannot move it. Property III implies that the gauge fixing only constrains the fields in $D(\Sigma)$. This reduces the problem to a special case of general question; how can the formalism of section \ref{subsec:gauge_subsystem} be adapted when gauge fixing conditions are imposed inside the subregion $\Sigma$?
%Gauge invariance of local measurements
\subsection{Gauge Fixed Subsystems}\label{subsec:gauge_fix}
This problem can be conveniently reformulated in a language that applies to any theory for which the locally gauge flows include diffeomorphisms. We restrict to an open set of $\mathcal{P}$ where we can impose a certain set of gauge fixing conditions. First, the  fixed slice $\mathfrak{C}$ is required to be Cauchy as before, and partitioned into regions $\Sigma$ and $\overline{\Sigma}$. This is not a serious restriction, since for any solution where $\mathfrak{C}$ is a Cauchy slice, there is an open neighbourhood in $\mathcal{P}$ where this is also true, and solutions outside of this neighbourhood can be gauge transformed into it. In other words it amounts to simply shrinking the considered open set. Furthermore, the location of the domain of dependence $D(\Sigma)$ is fixed. This amounts to gauge fixing the metric to be degenerate along certain hypersurfaces that end at $\partial\Sigma$, which are the putative null boundaries of the domains. This condition plays no role in defining the subsystem, but is useful for simplifying proofs. Afterwards, it may safely be dropped.
\par
Finally, there is the set of gauge fixing conditions that correspond to the constraining of $\partial\Sigma$ in the previous section, which are constructed out of the fields and their derivatives within $D(\Sigma)$. Together with the null boundary conditions, these are denoted $G^k=0$. They restrict the solutions to a hypersurface of nontrivial codimension within $\mathcal{P}$, albeit one for which there is a complete basis of transverse directions consisting of pure gauge transformations. Altogether, the constraints $G^k=0$ are assumed to eliminate gauge transformations that "move $\partial\Sigma$", and thereby dynamically determine $D(\Sigma)$. In precise mathematical terms, a finite residual gauge transformation is applied by integrating a flow $X_\lambda$ for a set of time dependent parameters $\lambda^a(t)$ such that the constraints are always preserved. The requirement is that for all allowed parameters, the fields within $D(\Sigma)$ after the transformation depend only on the fields within $D(\Sigma)$ beforehand, so the residual gauge transformations act as a Lie group on $D(\Sigma)$.
\par
The key point of the gauge fixing is that afterwards, there are functionals that are invariant under residual gauge transformations that were not invariant under the complete gauge group. This does \textit{not} mean that there are additional observables in the gauge fixed theory; rather, the original observables are equal to other functionals when restricted to the constraint surface. As is well known \cite{Wallace}, gauge fixing is a method of rewriting the gauge invariant observables in a way that is, hopefully, simpler. In particular, although the observables of the theory are natively nonlocal with respect to the background manifold, we can hope that some of them localize to $\Sigma$ after imposing $G^k=0$, and define an observable subsystem. Any such localized observables can be extended uniquely off the constraint surface to obtain a fully invariant expression, at least within the open set of consideration \cite{reference_frames,henneaux_teitelboim, Bianca, Chataignier}. 
\par
The central tools in section \ref{subsec:gauge_subsystem} were the distributions $\Ssp$ and $\Ssbp$. On the constraint surface, these restrict to distributions $\Sspg$ and $\Ssbpg$ respectively, which by Corollary \ref{causal_complement_gauge} exactly consist of the flows tangent to the constraint surface that restrict to gauge transformations within $D(\overline{\Sigma})$ and $D(\Sigma)$ respectively. Note that  Propositions \ref{complement_gauge} and \ref{inv_gauge} immediately descend to these distributions. For the former, it is because the flows in $\Sspg$ are equal to the flows in $\Ssp$ up to the addition of gauge vectors that annihilate $\Omega$, so the symplectic complement of $\Sspg$ is just the restriction of the symplectic complement of $\Ss$ to the constraint surface; in other words, $\Ssbpg$. For the latter, it is because we already know that the bracket of two vector fields in $\Ssp$ belongs to $\Ssp$, and the Lie bracket of vector fields tangent to the constraint surface is another vector field tangent to the surface. The statements for $\overline{\Sigma}$ hold symmetrically, as so far we have not employed the localization of the constraints to $\Sigma$.
\par
However, the symmetry is broken when considering how these flows restrict to the complementary subregion. The flows in $\Sspg$ can be split as before into $\eta+\zeta$, where $\eta$ is supported in $D(\Sigma)$ and $\zeta$ is gauge. Only the sum is required to preserve the constraints $G^k=0$; there is no condition on $\eta$ itself, since for any flow $\gamma$ tangent to $\mathcal{P}$ it is possible to choose $\zeta$ such that $\zeta+\gamma$ preserves the constraints. The flows $\Ssbpg$, however, can be split into $\chi+\zeta$, where $\chi$ is supported in $D(\overline{\Sigma})$ and $\zeta$ is gauge. The localization of $G^k$ implies that $\chi$ cannot affect their values, so $\delta G^k(\zeta)=0$ for all $k$ and $\chi$ is arbitrary. In other words, $\zeta$ is a residual gauge transformation. This implies that: 
\begin{prop}\label{prop:dress_int}
    The distribution $\Ssbpg$ is locally integrable.
\end{prop}
\begin{proof}
We can repeat the proof of Proposition \ref{gauge_int}, defining the leaves of the foliation by an equivalence relation. First, we identify solutions that agree on $D(\Sigma)$. Then the finite residual gauge transformations have a well-defined action on these equivalence classes, so we can identify their orbits as well. As before, this identification is implemented by paths, so the foliation is path-connected.
 \end{proof}   
This defines a phase space $\mathcal{P}(\Sigma)$ and an associated set of observables. Note that adding additional gauge fixing conditions $H^k$ that are also supported in $D(\Sigma)$ will not affect the construction of the subregion phase space, and will therefore yield the same observables, although they will admit new re-expressions.  As before, we can show that:
\begin{prop}\label{prop:Sspg}
    The set of observables generating flows belonging to $\Sspg$ is precisely the set of regular observables supported on $\Sigma$.
\end{prop}
\begin{proof}
   The proof is structurally identical to that of Proposition \ref{prop:Ssp}, making use of Proposition \ref{prop:dress_int} and the symplectic complementarity of $\Sspg$ and $\Ssbpg$.
\end{proof}

\begin{cor}
    The regular observables supported on $\Sigma$ form a closed Poisson algebra.
\end{cor}
\begin{proof}
    This follows from Proposition \ref{prop:Sspg} and the involutivity of $\Sspg$.
\end{proof}
Therefore the subregion $\Sigma$, or equivalently $D(\Sigma)$, can be regarded as a proper subsystem. Unlike before, these theorems do not apply to the observables localized in $\overline{\Sigma}$, because they are not necessarily invariant under $\Sspg$. Instead, they must be invariant under the direct sum of the restriction of $\Ss$ to the constraint surface and the residual infinitesimal gauge transformations, which is a proper subspace of $\Sspg$. 
\par
As in section \ref{subsec:gauge_subsystem}, the observables in $\Sigma$ generate field-dependent gauge transformations in the causal complement. However, these are \textit{not} necessarily residual gauge transformations, since for any $\eta\in\Ss$ there is a choice of $\zeta\in\mathcal{G}$ such that $\eta+\zeta$ preserves $G^k=0$, but neither must do so individually. Therefore the observables in $\Sigma$ do \textit{not} have to Poisson commute with observables in $\overline{\Sigma}$, which is consistent with the failure of microcausality expected in gravitational theories \cite{dress, obstruction_subsystem}. However, this failure is strongly constrained. There is a useful analogy with asymptotic symmetry generators, such as the ADM Hamiltonian. These exist at the boundary of spacetime, and might naively be expected to commute with fields in the bulk. However, they do not; instead they generate locally gauge flows in the bulk that have nontrivial boundary limits, and have nontrivial brackets with observables.
\par
There is a case in which microcausality of subsystems does hold, along the lines of \cite{reference_frames, Marolf_causality}: 
\begin{prop}
 Given two internally dressed subsystems that are spacelike separated in an open subset of $\mathcal{P}$, their algebras mutually Poisson commute.
\end{prop}
\begin{proof}
  Impose the gauge fixing constraints for both subsystems, which live in spacelike separate domains. The flow generated an observable in one region is gauge on the causal complement, but is otherwise unrestricted except for the condition that it preserves the constraints for the other subsystem. This can be achieved by switching off the gauge parameters in the intervening region, so that it leaves observables in the other subsystem unaffected. See figure \ref{fig:microcausal}.
\end{proof}
\begin{figure}
    \centering
    \includegraphics[scale=.2]{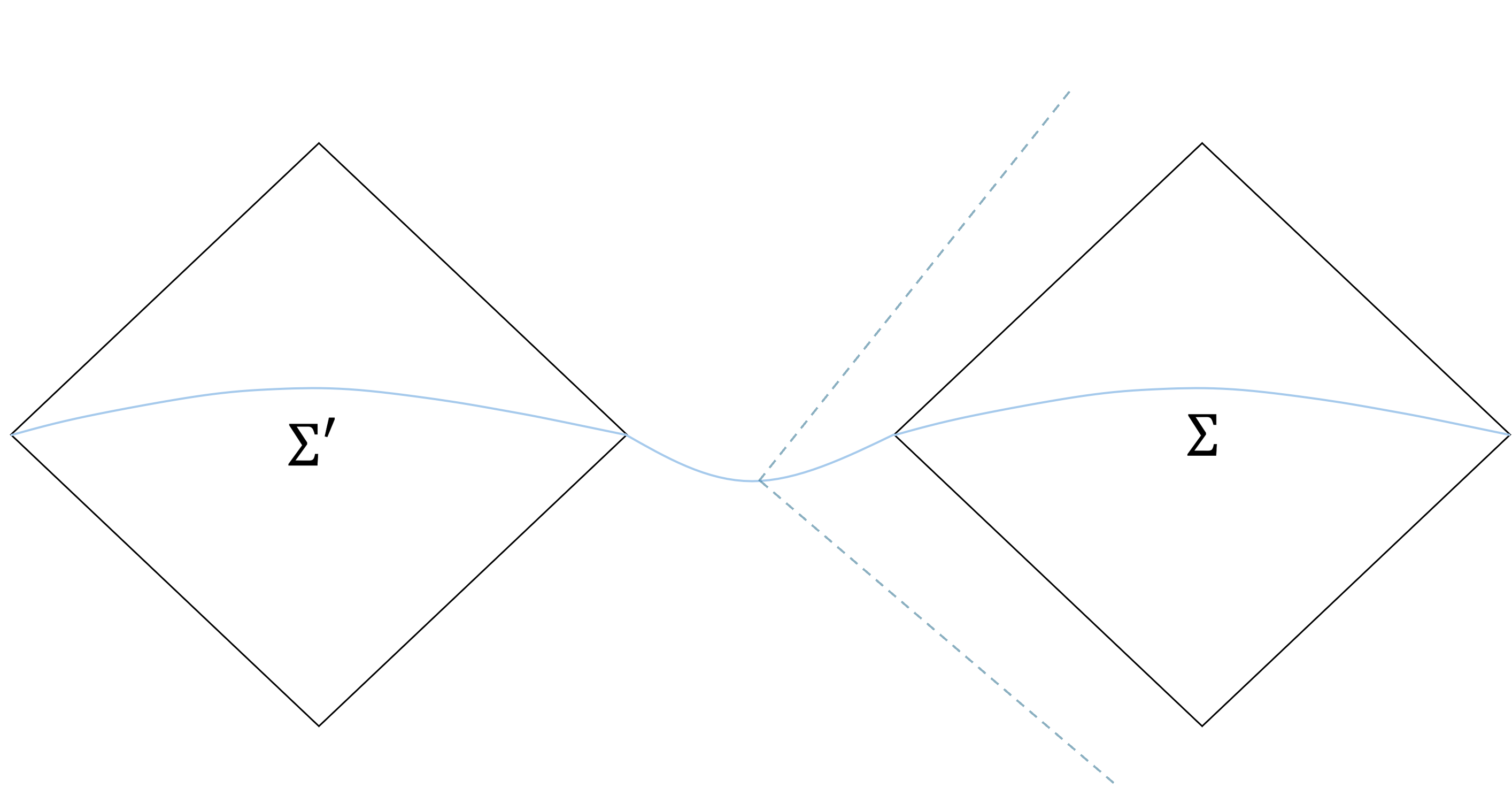}
    \caption{Observables supported in $D(\Sigma$) generate flows that are gauge on the causal complement, and can be switched off before reaching $D(\Sigma')$. The dotted lines indicate the boundary of the region where the flow is nonzero.}
    \label{fig:microcausal}
\end{figure}
As a "limiting case" note that if the dressing constraints for $D(\Sigma)$ are supported purely on the boundary (they consist of some local constraints on $\partial\Sigma$ as well as the null boundary conditions), the causal complement is \textit{also} an internally dressed subsystem. The algebras of these regions are easily seen to commute with each other.
%More detail?
\section{Some Examples}\label{sec:examples}
We will now discuss some examples that illustrate the general formalism. Because we are primarily interested in matters of principle, we will not do any involved computations of Poisson brackets, and will leave the treatment fairly schematic. The reader will be referred to the literature for more concrete details.
\subsection{Electromagnetism}
Although it is not a gravitational theory, it is instructive to consider the case of the Maxwell field coupled to a charged scalar. The theorems of \ref{subsec:gauge_fix} still apply. Consider the case where $\Sigma$ is a compact region, and we gauge fix the scalar field at one point in the interior to be real. Then choose a set of contours on the Cauchy slice radiating out from this point in all directions. Fix the pullback of $A_\mu$ along these lines to vanish within $\Sigma$. This is called a contour gauge \cite{contour}, and traces back to at least \cite{Mandelstam}. The observables within the region are exactly the same as they were initially, since there is no need to dynamically fix its location. A complete generating set is given by the electric and magnetic fields, as well as Wilson lines ending on scalar fields starting at the center and lying along the chosen contours and their time derivatives: \begin{equation}
  W[C;0,t]=\phi^*(C(t))e^{i\int_C A}\phi(C(0)).  
\end{equation} The gauge fixing conditions allow this to be rewritten as $W[C,t]=\phi^*(C(t))\phi(C(0))$. All of the regular observables in the region together form a Poisson algebra.

\par
Consider $W[C;-t,t]$, where $C$ passes through the center point and $C(-t),C(t)$ are on opposite sides of it. Although this observable is not originally supported on $\overline{\Sigma}$, the gauge fixing re-expresses it as 
\begin{equation}
    W[C,t]=\phi^*(C(t))e^{i\int_{C\setminus\Sigma} A}\phi(C(-t)), 
\end{equation}
which is. However, it cannot commute with all observables localized in $\Sigma$; for example, a smearing of the electric fields that intersects $C$. The flow generated by this electric field observable, as well as any other in $\Sigma$, must restrict to a gauge transformation on $\overline{\Sigma}$. This nontrivial action arises from the requirement of preserving the gauge constraints; computing it may be a useful exercise.
\subsection{Scalar Models}
On the gravitational side, the simplest conceivable examples of dressing are Z-models \cite{Z-model,obstruction_subsystem,Marolf_causality,locally_covariant_gravity}, which involve a set of $D$ scalar fields $X^i$ that are assumed to take monotonic profiles in an open set of $\mathcal{P}$. They can be used as coordinates and gauge fixed to take constant values (at least in a compact region). In the constrained theory, a partial Cauchy surface inside the region where $X^i$ are fixed defines an internally dressed subsystem. There are no residual gauge transformations, so all local fields represent observables. It is apparent that the minimal set of constraints required to define the subsystem are the values of $X^i$ at the corner $\partial\Sigma$, and thus the rest of the gauge fixing does not affect the set of localized observables. Since this condition is localized to the corner, the causal complement is also an internally dressed subsystem. As expected, the two algebras mutually commute, as was originally discussed in \cite{Marolf_causality}. A natural alternative dressing along the same lines can be obtained by considering two scalar fields, such that in some open subset of $\mathcal{P}$ their zero sets intersect at a compact, spacelike codimension-two surface. Taking the interior domain of dependence to be $D(\Sigma)$ defines an internally dressed subsystem with similar properties.
\par
A more physically realistic way to construct a complete set of coordinates is to use Brown-Kuchar dust fields \cite{brown_kuchar}, of which a nice covariant treatment can be found in \cite{reference_frames}. The basic idea is to use double the number of scalar fields to construct an action that describes a perfect fluid of zero pressure, as well as functions that parametrize its geodesic flowlines.
\subsection{Causal Patches}\label{subsec:causal_patch}
The most interesting examples of internal dressing are motivated by the analysis of black holes and holography. Consider a two-sided asymptotically AdS spacetime and a spatial subregion $R$ of one boundary. The causal wedge is defined as usual, by taking the boundary domain of dependence of $R$ and then considering the intersection of its bulk causal past and future \cite{causal_wedge, tasi}. For example, if one takes $R$ to be a Cauchy slice of one entire boundary, the wedge is the exterior of the wormhole. One can apply a diffeomorphism to fix the location of the light sheets bounding this region. It may not immediately be obvious that this diffeomorphism will be gauge; to see this, note that the AdS boundary conditions at the conformal boundary are a limiting case of Dirichlet conditions at a finite boundary. This means that the geometry of the light sheets is fixed to first order away from the boundary, and the diffeomorphism needed to fix them completely in place can vanish at the boundary. The surface charges $q_\xi$ of Einstein gravity vanish at Dirichlet boundaries when the vector $\xi$ vanishes \cite{harlow}, so by Proposition \ref{prop:local_gauge} the insertion of such a diffeomorphism flow into the symplectic form vanishes, and it is gauge. This discussion is completely schematic, and requires holographic renormalization to be made rigorous. However, I expect that it is essentially correct and that the causal wedge can thus be defined as an internally dressed subsystem. The causal complement of the causal wedge is not an internally dressed subsystem; however, the causal wedge of $\overline{R}$ is, and is always spacelike separated from the causal wedge of $R$ \cite{Hubeny_causal}. The classical algebras of observables of the two causal wedges must therefore commute; however, since the wedges do not cover a bulk Cauchy slice, it seems that together they do not generate the complete algebra of observables of the theory. \par
More interesting in the context of holography, however are \textit{entanglement wedges}. Supposing asymptotically AdS boundary conditions, in the classical limit the entanglement wedge of $R$ is defined as the bulk domain of dependence of a partial Cauchy surface $\Sigma$ enclosed by $R$ and a locally extremal homologous bulk codimension-2 surface \cite{entanglement_wedge, Harlow_wedge}, called the HRT surface \cite{HRT}. There is always at least one such surface; for simplicity, we can restrict to the open set where there is only one, although it is easy to generalize when there are multiple. As before, it is expected that the location of this surface and its associated light sheets can be gauge fixed in the bulk, since the first order behaviour off of the conformal boundary is fixed. This time the entanglement wedge of $\overline{R}$ is equal to the causal complement of the entanglement wedge of $R$, and both are internally dressed subsystems. This is a special case of the property noted at the end of section \ref{subsec:gauge_fix}, where the causal complement of an internally dressed subsystem defined by conditions that are localized on its boundary is also internally dressed.
\par
There are many other examples that one could consider, such as the timelike envelopes of dynamical worldlines \cite{Witten_ARO, Witten_background_independent,Kaplan_dS} or regions constructed by shooting in geodesics from the boundary. As long as these constructions are carried out in such a way that the dressing elements are contained in the final domain of dependence, they will define consistent subsystems with closed Poisson algebras. 
\section{Discussion}
\subsection{Factorization and Subregion Independence}\label{subsec:factor}
Now we will consider some additional problems to which our methods lend useful perspectives. The first is the question of the independence of subregions in classical gravity. We have proven a result about internally dressed, spacelike separated subsystems that is a good starting point. However, there are other interesting viewpoints that one could take. Recently \cite{Folkestad} Folkestad discussed the question of whether, perturbatively around a fixed background, spacelike separated subregions obey a classical split property where perturbations made to one subregion can be extended smoothly to the complement in such a way that they vanish on the other subregion. By considering several examples, he found that this is essentially the case around generic asymmetric backgrounds \footnote{This is suggestive of a relationship with Proposition \ref{prop:local_gauge_noether}, which holds exactly for backgrounds that have no local symmetries. Perhaps this tool can provide an exact formulation and proof of Folkestad's conceptions.}, and took this as evidence that independent localized algebras of observables may exist. We have sidestepped this issue with the internal dressing condition; as long as one only considers perturbations generated by internally dressed observables, these are gauge in the causal complement and can be switched off smoothly in such a way that they vanish in the other subsystem. The question raised by Folkestad is still interesting in its own right, especially if it is discussed in terms of whether arbitrary allowed perturbations to the observables in one dressed subsystem can be extended to leave invariant the observables in the other subsystem. We do not claim an answer to this, but it would be interesting to pursue in future work.
\par
Another question of interest is whether the phase space factors into the phase spaces of complementary subregions. This was studied in \cite{Riello1,Riello2,Riello3,Riello4} within the context of the linearization of Yang-Mills theory. To be mathematically precise: for a particular solution, consider the Cauchy surface $\mathfrak{C}$ partitioned into $\Sigma$ and $\overline{\Sigma}$. Do the linearized observables contained in either region determine the values of all linearized observables?  Equivalently, is a smooth linearized solution $\gamma$ determined up to gauge by how it affects subregion observables? The is answered in the affirmative in the as the "general gluing theorem" of \cite{Riello3}. This states that for Yang-Mills $\gamma$ is determined up to the addition of a surface symmetry, as long as the background solution is not invariant under any locally gauge flows. The proof involved various technical, model-specific computations. Remarkably, however, we can provide an alternate proof using the results of section \ref{subsec:gauge_character}, for any theory satisfying Assumptions \ref{causal_1} and \ref{causal_2}! To see this, note that the effect of $\gamma$ on regular linearized observables localized in $\Sigma$ is equivalently described by by $\Omega_\Sigma(\gamma,\eta)$ for all $\eta\in\Ss$. By Proposition \ref{prop:local_gauge} determines $\gamma$ within $D(\Sigma)$ up to the addition of a locally gauge flow. Likewise $\gamma$ is determined within $\overline{\Sigma}$ up to the addition of a locally gauge flow. In total this means that the effect of $\gamma$ on subregion observables determines it everywhere up to the addition of a locally gauge flow, which can be a nontrivial surface symmetry. As long as the background is not left invariant by any gauge symmetry (in the language of \cite{Riello3}, it is irreducible), this will be a Noetherian symmetry of the theory.
\par
The main conclusion is that the nonfactorization of this linear gauge theory phase space is, with great generality, entirely due to the existence of surface symmetries. Note that we are \textit{not} referring to subregion surface symmetries as in \cite{subsystem}, but rather to symmetries of the entire spacetime. This is compatible with the perspective of \cite{harlow_factorization}, that the nonfactorization of the Hilbert space of electromagnetism in a two-sided asymptotically AdS wormhole can be resolved within effective field theory if the gauge symmetry is emergent. We see that at the level of the phase space, factorization is \textit{only} possible if either at the fundamental level there really isn't a gauge symmetry, or that somehow the boundary conditions are stronger than expected, so that there are no nontrivial surface symmetries. It would be nice to extend this analysis to the full nonlinear phase spaces and observable algebras defined in sections \ref{subsec:gauge_subsystem} and \ref{subsec:gauge_fix}, but unfortunately this is beyond the tools that we've developed here.
\subsection{Surface Charges and Kink Transforms}\label{subsec:surface}
Surface symmetries of subregions and their generating charges have recently attracted a lot of interest, following \cite{subsystem}. We would like to see how these fit into the framework of this paper.  For now, let us put aside the question of gravitational dressing and work with pure Maxwell theory. The symplectic form is given by 
\begin{equation}
    \Omega = \int_\mathfrak{C} \delta A\wedge\star\delta F.  
\end{equation}
Consider a gauge transformation $\delta A= d\rho$, $\delta \star F=0$. Define a singular flow $\tau$ given by $I_\tau\delta A=1_{\Sigma}df$ and $I_\tau \delta\star F=0$ on $\mathfrak{C}$, where $1_{\Sigma}$ is the indicator function for the compact region $\Sigma$. Note that $\tau$ preserves the Gauss' law constraint $d\star F=0$ on the Cauchy slice. Inserting this flow into $\Omega$ yields 
\begin{equation}
    -I_\tau\Omega = \int_\Sigma -df\wedge \star\delta F=  \int_\Sigma -[d(f\star\delta F)-f\delta(\cancel{d\star F})]=-\int_{\partial\Sigma} \delta(f\star F),
\end{equation}
which is minus the variation of the electric flux through $\partial\Sigma$ smeared against $f$. This is a gauge invariant singular generator of the flow $\tau$, which can be propagated throughout $M$ as an on-shell distributional solution of Maxwell's equations. Based on a lattice field-theoretic viewpoint, this is often referred to as being a central observable .

\par
Alternatively one can consider the flow $\upsilon$ given by $I_\upsilon\delta A=d(1_{\Sigma}f)$ and $I_\upsilon \delta F=0$. Inserting this into $\Omega$,
\begin{equation}
    -I_\upsilon\Omega = \int_\mathfrak{C} -d(1_{\Sigma}f)\wedge\star\delta F=  \int_\mathfrak{C} -[d(1_{\Sigma}f\wedge\star\delta F)-1_{\Sigma}f\wedge\delta(\cancel{d\star F})]=0
\end{equation}
by Stokes' theorem for currents. We see that multiplying the gauge transformation of $A$ by an indicator function yielded a nontrivial surface charge that generates a flow that equals the gauge transformation within $\Sigma$ and vanishes in $\bar{\Sigma}$. Multiplying the gauge parameter \textit{itself} by an indicator function yields a distributional gauge transformation that still annihilates the symplectic form. The two flows differ by a delta function in the perturbation of $A$, localized at $\partial\Sigma$, which exactly cancels out the electric flux contribution from $\tau$ in $\Omega$. 
\par
This is an example of the general structure of \textit{subregion surface symmetries}, which are singular flows that restrict to a particular gauge transformation in a subregion $\Sigma$ and vanish on the complement $\overline{\Sigma}$. Two such flows differ on the Cauchy slice only by their singularity structure at $\partial\Sigma$, but this makes all the difference. There is always a choice of proper singular gauge transformation with these restrictions, but by Proposition \ref{prop:local_gauge} we know that these will always yield zero upon insertion into $\Omega$. The other subregion surface symmetries will yield nonzero results that can be traced back to some delta function terms localized at $\partial\Sigma$. 
\par
At a particular point in $\mathcal{P}$, the singular generator for a particular surface symmetry is the 1-form obtained by inserting it into $\Omega$. Note that if the flow is on-shell, the generator must be gauge invariant. For example, consider the analogue of $\tau$ for Yang-Mills theory, which is the product of an indicator function with the variation of $A$ under a gauge transformation with Lie algebra valued parameter $g$. The generator $-I_\tau \Omega$ is now the variation of a smearing of the colour electric flux through $\partial\Sigma$, as computed in \cite{subsystem}. This is not a gauge invariant quantity, and indeed $\tau$ does not preserve the nonabelian Gauss' law constraints $D\star F=0$. To see this, note that $I_{\tau-\upsilon} \delta A=1_\Sigma Dg-D(1_\Sigma g)=-g(d 1_\Sigma)$, so 
\begin{equation}
   I_\tau \delta (D\star F)= \cancel{I_\upsilon \delta (D\star F)}+ I_{\tau-\upsilon}D\star F= D\star D(-g(d 1_\Sigma))+[-g(d 1_\Sigma)\wedge\star F],
\end{equation}
which is made of several terms that are delta localized on $\partial\Sigma$ and do not cancel.
\par
Although there is not a unique singular flow and surface charge associated with a specific gauge transformation and restricted subregion, there may still be a choice that is particularly natural. This is the case for the kink transform in Einstein gravity \cite{connes, marolf_kaplan}, which is one way to formalize the notion of a "one-sided boost", which is particularly efficient in the distributional approach that we have introduced. Consider a gauge transformation that acts like a local Lorentz boost (normalized such that the local boost rate is $2\pi$ i.e. continuing to Euclidean signature would give a full rotation in unit parameter time)  near a compact codimension-two surface $\partial\Sigma$ \footnote{The authors of \cite{marolf_kaplan} also consider the noncompact case, in which the following discussion is altered by boundary conditions; the area will generate a flow that differs from the kink transform.}, by which we mean a diffeomorphism vector field $\xi$ that vanishes at $\partial\Sigma$ and preserves all components of the metric there. The kink transform flow $\kappa$ is obtained by setting $\delta g_{ab}=1_{R}\mathcal{L}_\xi g_{ab}$, where $R$ is a spacetime region containing $D(\Sigma)$ such that $\partial\Sigma\subset\partial R$. To compute the surface charge for this flow, take the usual covariant choice for the symplectic current of Einstein gravity as given in \cite{harlow}, which involves only the variation of the metric and its first derivatives. Since $\mathcal{L}_\xi g_{ab}$ vanishes at $\partial\Sigma$, neither $\delta g_{ab}(\kappa)$ nor $\nabla (\delta g_{ab}(\kappa))$ contains delta function singularities at $\partial\Sigma$. Therefore $-I_\kappa \Omega$ restricts to an integral over $\Sigma$ that simplifies to $\int_{\partial\Sigma} q_\xi$. This in turn is well known to be equal to $\frac{\delta A}{4G_N}$ for the chosen symplectic current \cite{subsystem, marolf_kaplan}. It is important to realize that, although we picked a specific way to fix the bulk exact ambiguity in order to run this argument, the final result is independent of that choice, being defined by the current integral of $-I_\kappa\omega$ as discussed in section \ref{subsec:subsystem}. 
\begin{figure}
    \centering
    \includegraphics[scale=.2]{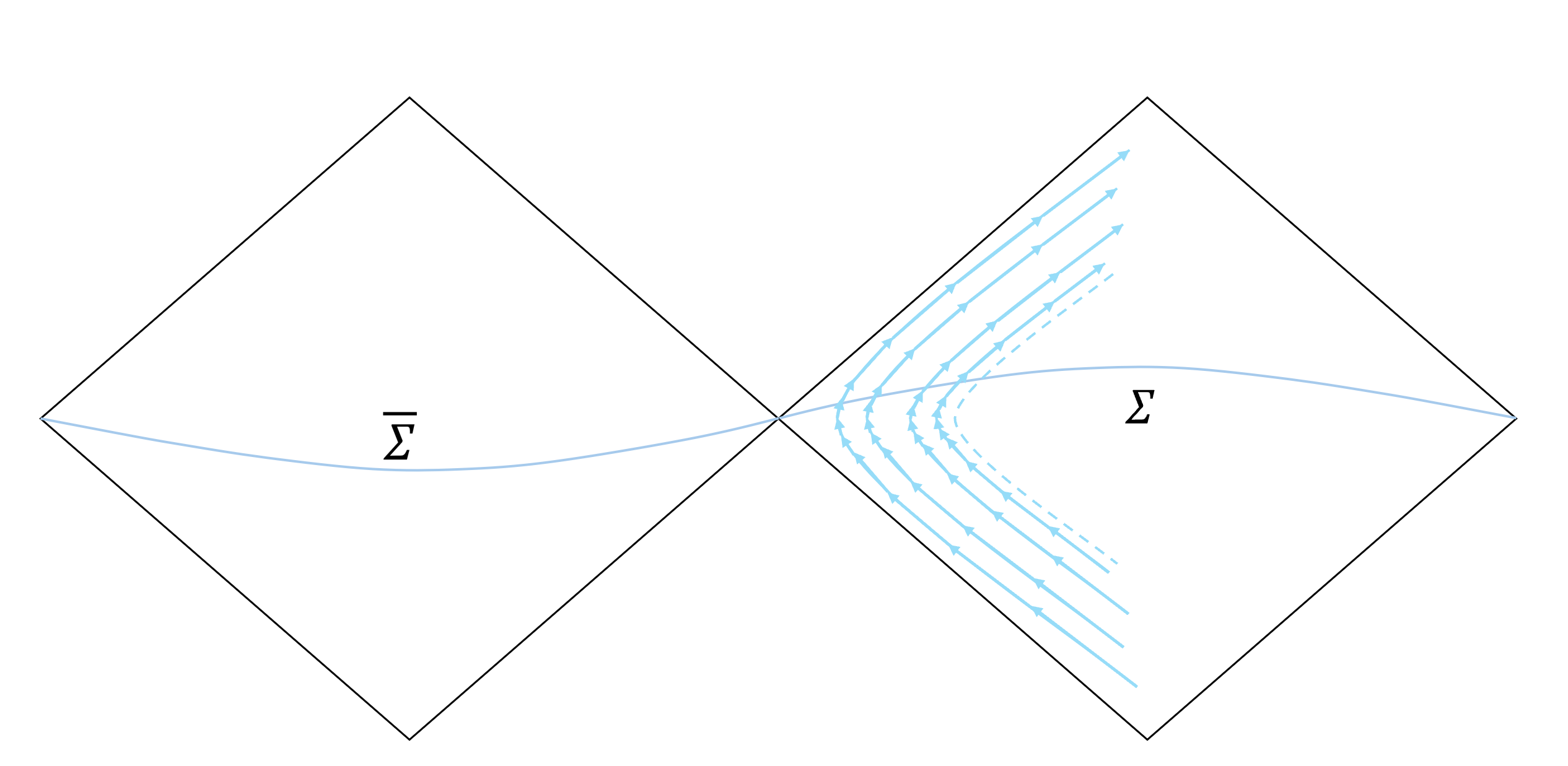}
    \caption{The effect of the kink transform is illustrated. In $D(\Sigma)$ it looks like a boost near $\Sigma$, but is switched off to vanish by some point demarcated by the dotted line,}
    \label{fig:kink}
\end{figure}

\par
One application of the kink transform is that it provides a form of Wald's computation of the boost Noether charge for bifurcate black hole horizons that is \textit{not} subject to the ambiguities noted in \cite{JKM}. This works precisely because the flow is taken to be distributional and then integrated \textit{across} the bifurcation surface. One could object that there is a compensating ambiguity in how the singularity structure of the flow is determined. However the definition that multiplies the action of the boost on the metric is, at the very least, a fairly natural ansatz. It would be interesting if this technique could be applied to higher curvature theories, where there are holographic and thermodynamic arguments in favour of the Dong-Wall form \cite{dong_entropy, wall_entropy} as the appropriate surface charge associated with black hole entropy. There are higher derivative subtleties involving how the boost diffeomorphism and indicator function are to be defined, which renders this somewhat difficult. This line of work will be pursued in the future.
\par
The area of $\partial\Sigma$ is not gauge invariant at first order around an arbitrary solution, so the flow $\kappa$ is not generally on-shell. When $\partial\Sigma$ is \textit{extremal}, however, the area is infinitesimally gauge invariant. One can check that $\kappa$ does not satisfy the linearized equations generally, but as shown in \cite{kink} does when the surface is extremal. Rather than doing this in full generality, let us focus on one null hypersurface passing through $\partial\Sigma$ and consider the Raychaudhuri equation. This takes the form 
\begin{equation}
    \dot{\theta}+\frac{1}{D-2}\theta^2+\sigma^2+...=0,
\end{equation}
where $\theta$  is the expansion of the affinely parametrized horizon null generators, $\sigma$ is the shear, and $\dot{\theta}$ denotes the covariant derivative along generators. The kink transform acts as a boost above the cut $\partial\Sigma$ and trivially below. Local boosts scale the expansion by $\delta_\xi \theta = \lambda \theta$ so $\dot{\theta}$ picks up a delta function contribution. The other terms have step function discontinuities, but these are irrelevant; the delta function alone produces a violation of the Einstein equation, effectively introducing a "stress tensor" on $\partial\Sigma$ that is not cancelled by any other terms. Note that when $\theta$ vanishes at the surface, this contribution goes away. Together with the transverse Raychadhuri equation, this confirms that $\partial\Sigma$ \textit{must} be extremal for $\kappa$ to be on-shell. 
\begin{comment}
The true distributional gauge transformation that is obtained by multiplying $\xi$ by the indicator function is of course on-shell; one can check that the delta functions in the derivative of the expansion are cancelled by delta functions in the Christoffel sybmols that appear with in the inaffinity.
\end{comment}
\par
These arguments were made in the absence of gauge fixing conditions. If the gauge is fixed to eliminate all diffeomorphisms that move $\partial\Sigma$, then the area is automatically invariant. However the deformation by the constraints means that it will not in general generate the kink transform. The exception is when the extremal surface condition is used to define the region, as described in \ref{subsec:causal_patch}. This guarantees that the subsystem contains a singular generator for the kink transform. Compare this to how the extremal surface condition arises in \cite{LM} from the requirement that a certain linearized perturbation to a Euclidean metric preserves the Einstein equations. This perturbation is defined by the insertion of a Euclidean conical defect, which is closely related to a Lorentzian boost \cite{euclidean_boost}. It is possible that there is some deep connection between these two ideas, which remains to be fully explored.
\subsection{Prospects for Quantization}\label{subsec:quantum}
We have discovered a way of making sense of subsystems in classical gravitational theories that is based on the consideration that observables belonging to a subsystem should form a closed Poisson algebra. This is however not really a classical condition, since the bracket structure is not an independently observable classical quantity. All of the measurable information in a classical theory is encapsulated by the set of solutions to the equations of motion. Although Peierls brackets are natural to consider in the context of an action principle, their proper interpretation is as the leading order contribution to quantum commutators in an $\hbar$ expansion. We would like to understand how to construct the algebra of operators at higher orders in $\hbar$, incorporating quantum effects, and how the notion of subsystem algebras extends.
\par
A different but related problem is studied by the framework of \textit{perturbative algebraic quantum field theory} (pAQFT) \cite{locally_covariant_gravity, pAQFT}, which is essentially a mathematically rigorous application of the interaction picture to quantum field theory. Starting with a free theory defined on general globally hyperbolic metrics, the authors perturbatively turn on compactly supported local interactions. Operators in the deformed theory are described as power series in the coupling of operators belonging to the free theory in the past of the interacting region. Employing a BV-BRST framework and carefully defining the perturbative renormalization in a convenient gauge, this construction can be carried out for general gauge theories, including gravity\footnote{The renormalization process will generally induce infinitely many irrelevant couplings. These are, as usual in effective field theory, interpreted as a tower of increasingly weak corrections to the low-energy physics.}. This will involve unobservable fields, including gravitons, matter, and appropriate ghosts and antighosts. Mathematically, this defines a perturbative $^{*}$-product, where products of operators in the interacting theory are defined as power series of operators in the free theory. Formally inverting the map to the interacting theory would yield an interacting operator, and could in principle yield local operator algebras. Observables would be identified as those operators that are invariant under the BRST symmetry. However, the adiabatic switching of interactions introduces some difficulty in the interpretation of this construction. In particular the formal inversion may be undefined, so that one doesn't properly obtain an algebra for the interacting theory \footnote{I thank Kasia Rejzner for explaining this point.}.
\par
In an upcoming paper \cite{algebra_quantum_gravity} I will take a different point of view. Rather than a $\sqrt{G_N}$ expansion around a particular background, an $\hbar$ expansion can be defined without reference to a particular background, as an expansion around the \textit{classical theory}. This constitutes a \textit{deformation quantization} of the Peierls bracket\footnote{This problem was also considered by Hawkins and Rejzner \cite{star_product_interacting}, but they still employed a form of switching to a free theory as in the rest of the pAQFT literature. There are also some obstacles with renormalization of the algebra. The path integral approach turns out to avoid these issues.}. The basic idea is to consider a Schwinger-Keldysh path integral defined by some action, with arbitrary ingoing boundary conditions. Operators are defined as time-ordered functionals of the fields, which can be inserted into the integral. The Schwinger-Dyson equations associated to this path integral formally define relationships between operators. It turns out that this allows the construction of an algebra of operators, and of states on this algebra, where both are defined as power series in $\hbar$. In order for the states to be well-defined, the action must include appropriate renormalization counterterms. This approach can be applied to general quantum field theories, including gravitational effective field theories, and can be used to study interesting semiclassical quantities such as partition functions. 
\par
One may wonder why the subsystem algebras are formulated as $\star$-algebras in $\hbar$ rather than von Neumann algebras, as is standard in algebraic quantum field theory \cite{Haag}. One reason is that the expansion in $\hbar$ does not seem to be compatible with a notion of a perturbative Hilbert space, in which any kind of completion could be taken. The other problem is that thinking of subsystem algebraic closure in von Neumann terms can be misleading. Given any set of operators acting on a Hilbert space, one can take a completion via the double commutant, which is automatically a von Neumann algebra with closed product relations. This obscures the reasons that algebras associated to causally closed (and internally dressed in the gravitational case) regions are special. Classically these regions are closed under Poisson brackets. In the path integral framework that will be described in \cite{algebra_quantum_gravity}, they will be closed under $^{*}$-product. Put it another way, the product of any two time-ordered operators in such a region can be rewritten in terms of time-ordered operators in that region. This may be a conceptual property that clarifies how subsystems should be thought of in quantum theories more broadly, beyond the semiclassical approximation.

\textbf{Acknowledgements}
I thank Eanna Flanagan, Venkatesa Chandrasekaran, Robin Oberfrank, Laurent Freidel, Daniel Harlow, Juan Maldacena, Gautam Satishchandran, and Kasia Rejzner for useful interactions. I am especially grateful to Rodrigo Andrade e Silva, Antony Speranza, Ted Jacobson, Joshua Kirklin, and Aldo Riello for deep discussions, and to the hospitality of the Perimeter Institute for Theoretical Physics, where part of this work was done. GPT-6 Astra assisted in revising some technical points. This study was supported in part by US National Science Foundation grants PHY-2012139
and PHY-2309634. Research at Perimeter Institute
is supported in part by the Government of Canada through the Department of Innovation,
Science and Economic Development and by the Province of Ontario through the Ministry of
Colleges and Universities.
\appendix
\section{Locally Gauge Flows and Noether Identities}\label{app:gauge}
In this Appendix we will prove Propositions \ref{prop:local_gauge} and \ref{prop:local_gauge_noether}, relying on Assumptions \ref{causal_1} and \ref{causal_2}. Before doing this, we review some properties of the linearized equations of motion $\delta E_a$. In this form they are written as the field space exterior derivative of the full nonlinear equations of motion obtained by functionally differentiating the action. However, they can also be obtained by expanding the action around a solution as $\phi_0+\delta\phi$ and taking the leading term, which is quadratic in $\delta\phi$. By taking functional derivatives with respect to $\delta\phi$, one recovers $\delta E_a$ as the equations of motion for $\delta\phi$. Noether's second theorem implies that the Noetherian gauge symmetries of this quadratic action correspond to nontrivial Noether identities among the $\delta E_a$, as well as the converse. The symmetries can be reverse engineered from the Noether identities as follows: write a Noether identity as $\Delta(\delta E_a)=0$ and weight this by some arbitrary functions $\lambda^a$. Integrating this by parts, 
\begin{equation}
	\Delta(\delta E_a)\lambda^a=-\delta E_a \Delta^*(\lambda^a)+dS'=0,
\end{equation}
where $\Delta^*$ denotes the formal adjoint of the differential operator $\Delta$ and $S'$ is some field space 1-form locally constructed out of $\delta E_a$. Taking $X_\lambda^a=\Delta^*(\lambda^a)$, note that this implies that the flow $X_\lambda^a \frac{\delta}{\delta\phi^a}$ changes the quadratic action by an exact term, and so is a Noetherian symmetry, and is gauge because $\lambda$ can be varied freely. All of these are on-shell (they satisfy $\delta E_a(X_\lambda)=0$). A simple way to prove this is to recall that $d\omega=\delta \phi^a\delta E_a$, and then note that 
\begin{equation}\label{gauge_insert}
d(I_{X_\lambda}\omega)=X_\lambda^a\delta E_a-\delta\phi^a \delta E_a(X_\lambda)= dS'-\delta\phi^a \delta E_a(X_\lambda).
\end{equation}
 Rearranging this we find that $\delta\phi^a \delta E_a(X_\lambda)$ is an exact term, which is not possible unless $\delta E_a(X_\lambda)=0$ because $\delta\phi^a$ is arbitrary. Note that this makes the right hand side of equation \eqref{gauge_insert} reduce to $dS'$.
\par
Now we will prove Proposition \ref{prop:local_gauge_noether}, which we restate as:
\begin{prop}
	For a given solution in $\mathcal{P}$ and an open neighbourhood $U$, the locally gauge flows are the same as the flows $X_\lambda$.
\end{prop}	

\begin{proof}
	 Suppose that $X$ is locally gauge within $U$, so that $I_X\omega + C=-dq$, where $C$ is locally constructed from $\delta E_a$. Taking the exterior derivative of both sides, we obtain a Noether identity:
	 \begin{equation}
	 d(I_{X}\omega+C)=X^a\delta E_a-\delta\phi^a\cancel{\delta E_a(X)}+dC=0.
	 \end{equation}
	 Therefore $X^a\delta E_a$ is equal to an exact term for off-shell $\delta\phi$, and is a Noetherian gauge symmetry for the linearized equations. \par
	 Now consider a symmetry $X_\lambda$. We already have that $d(I_{X_\lambda}\omega-S')=0$. Applying the Algebraic Poincare Lemma \ref{prop:apl}, we find that $I_{X_\lambda}\omega-S'=-dq$ for some local $q$. Since $S'$ is a linear combination of $\delta E_a$, it takes the role of $-C$, verifying that $X_\lambda$ is locally gauge. Note that since the forms involved are purely linear in the field space sections, which in this case are the linearized perturbations $\delta\phi$, this only requires the weaker form of the Lemma that appears as Lemma 1 in \cite{Wald_locally_constructed}. Since this version is proven only by integration-by-parts style manipulations, it applies just as well when the coefficients are distributional. In particular it applies when $\lambda^a$ are distributions rather than smooth functions.
\end{proof}
\begin{figure}
    \centering
    \includegraphics[scale=.2]{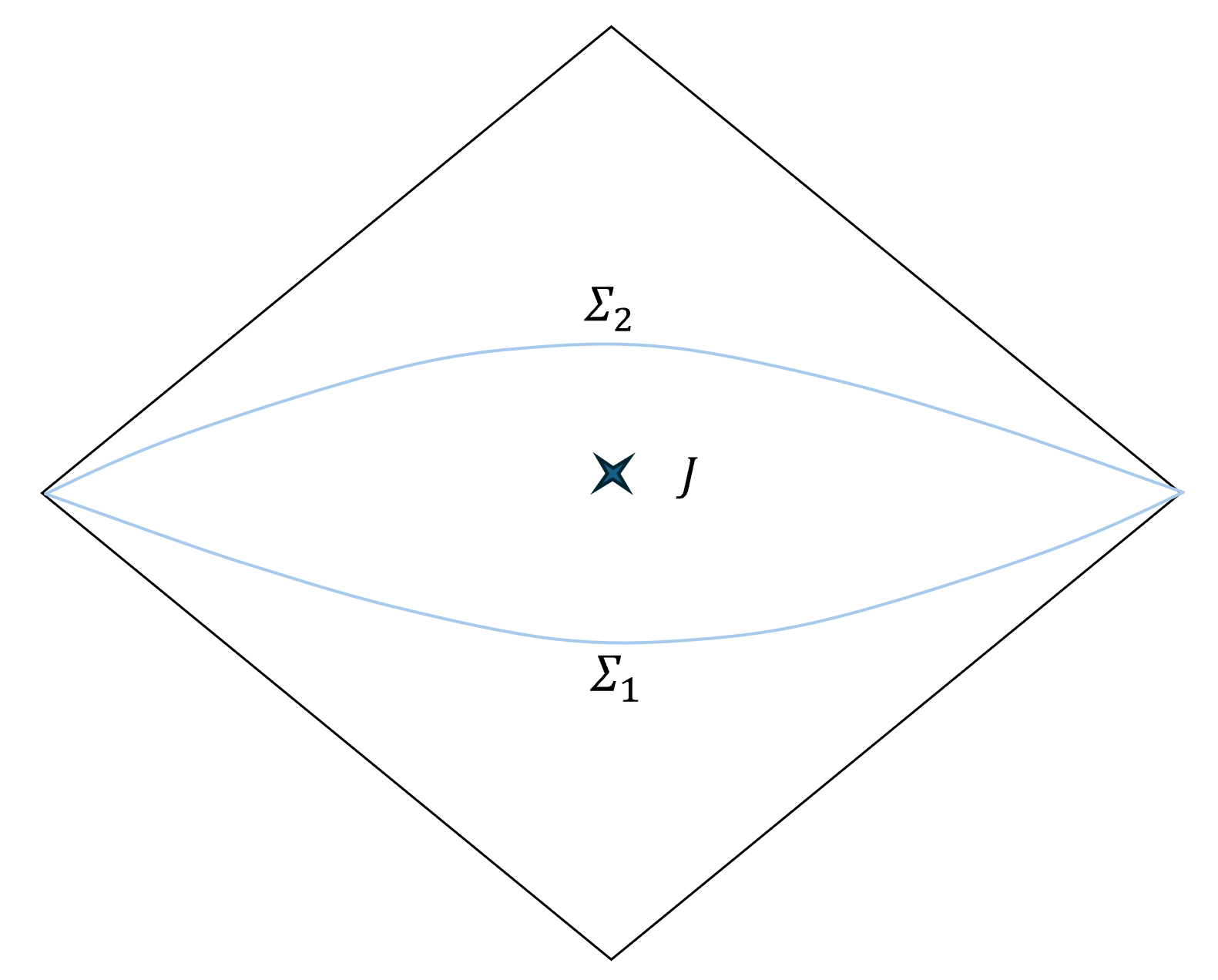}
    \caption{The surfaces $\Sigma_1$ and $\Sigma_2$. The region between them is where the interpolating flow is sourced; above $\Sigma_1$ it is equal to $\eta$ and below $\Sigma_1$, it vanishes.}
    \label{fig:interpolate}
\end{figure}
Proposition \ref{prop:local_gauge} is more subtle, as it requires a careful application of Assumption \ref{causal_1}. It states
\begin{prop}
	For a partial Cauchy surface $\Sigma$, any (distributional) on-shell flow $X$ such that $\Omega_\Sigma(X,\eta)=0$ for all flows $\eta$ that go to zero smoothly at $\partial\Sigma$ must be locally gauge within $D(\Sigma)$.
\end{prop}
\begin{proof}
	Consider homologous partial Cauchy surfaces $\Sigma_1$ and $\Sigma_2$, where the latter is located to the future of the former.  Both have the same domain of dependence (see figure \ref{fig:interpolate}). Consider a flow $\eta$ that vanishes smoothly towards the corner. By assumption,  $\Omega_{\Sigma_1}(X,\eta)=0$. Since $\Omega_{\Sigma_1}=\Omega_{\Sigma_2}$ for on-shell flows, $\Omega_{\Sigma_2}(X,\eta)=0$ as well. Alternatively, we can consider an off-shell flow $\tilde{\eta}$ that interpolates between the zero solution before $\Sigma_1$ and $\eta$ after $\Sigma_2$. We compute
    \begin{equation}\label{interpolate}
        \Omega_{\Sigma_2}(X,\tilde{\eta})- \Omega_{\Sigma_1}(X,\tilde{\eta})=\int_{\Sigma_1}^{\Sigma_2} d\omega(X,\tilde{\eta})=\int_{\Sigma_1}^{\Sigma_2} X^a\delta E_a(\tilde{\eta}), 
    \end{equation}
    which we already know vanishes. The flow $\tilde{\eta}$ can be viewed as a sourced solution of the linearized equations of motion where the quantity $\delta E_a(\tilde{\eta})$ is the required driving source, which we will denote as $J_a$. Note that $J_a$ must satisfy the same Noether identities $\Delta(J)=0$ as the linearized equations.
    \par
    The key observation is that, as a consequence of Assumption \ref{causal_1}, $\eta$ can be specified by an arbitrary data set on $\Sigma_2$ that is only constrained by some components of $\delta E_a$ and possibly derivatives thereof. This means that $\tilde{\eta}$ is only constrained by the condition that $J_a$ is smooth, and vanishes above $\Sigma_2$ and below $\Sigma_1$. Again by Assumption \ref{causal_1}, \textit{any} such  configuration of $J_a$ satisfying the Noether identities can be constructed by appropriately choosing $\eta$ and the interpolating configuration. We can now rewrite equation \eqref{interpolate} as 
    \begin{equation}
        \int_{\Sigma_1}^{\Sigma_2} X^a J_a =0.
    \end{equation}
Define the distribution $X[J]=\int_{\Sigma_1}^{\Sigma_2} X^a J_a$ for \textit{arbitrary} test sources, which must vanish when evaluated on test sources that vanish towards the boundary and satisfy the identities. By linearity, the value of $X$ for an arbitrary test source is determined only by the failure of the Noether identities in the interior, which we can write as $\Delta^j J=\mathcal{N}^j$ by treating $\mathcal{N}^j$ as "sources" for the identities, and the value of $J$ and its derivatives at the boundary. This means that we can write
\begin{equation}
    X[J]=\tilde{X}[\mathcal{N}]+\int dY[J],
\end{equation}
where the $\tilde{X}$ is a distribution acting on the Noether sources $\mathcal{N}$ and $dY$ is a total derivative that integrates out to a surface term (note that $Y$ itself is a degree $D-1$ de Rham current). When $J$ is again taken "on-shell" of the Noether identities, the first term vanishes, and we are left with the local identity
\begin{equation}
    d\omega(X,\gamma)=dY[J(\gamma)],
\end{equation}
where $\gamma$ is an arbitrary off-shell flow. Applying the Algebraic Poincare Lemma again, we deduce that between $\Sigma_1$ and $\Sigma_2$, the flow satisfies 
\begin{equation}
    I_X \omega = Y[\delta E_a]-dq,
\end{equation}
which is the locally gauge condition with $-Y=C$. Pushing $\Sigma_2$ and $\Sigma_2$ to the boundaries of the domain of dependence, we find that this holds throughout $D(\Sigma)$.
 \end{proof}
\begin{comment}
    It is instructive to briefly examine why this fails to hold outside of the domain of dependence, as discussed in \ref{subsec:gauge_character} and depicted in figure \label{fig:causal_example}. In the case where $\Sigma_2$ is not a partial Cauchy surface, the linearized fields on it will need to satisfy further constraints beyond those arising locally from the components of $\delta E_a$ in order to be propagated sourcelessly back to $\Sigma_1$.
\end{comment}

\bibliography{refs}

\end{document}